\documentclass{article}

\usepackage{graphicx} 
\usepackage{adjustbox} 
\usepackage{enumerate} 
\usepackage{geometry} 
\usepackage{amsmath} 
\usepackage{amssymb} 
\usepackage[mathletters]{ucs} 
\usepackage[utf8x]{inputenc} 
\usepackage{fancyvrb} 
\usepackage{grffile} 
\usepackage{hyperref}
\usepackage{authblk}
\usepackage{booktabs}  
    
\usepackage{mathtools}
\usepackage{soul}
\usepackage{float}

\usepackage[format=plain, labelfont={it}, textfont=it]{caption}

\title{No arbitrage SVI}

\author[1]{Claude Martini\thanks{cmartini@zeliade.com}}
\author[1,2]{Arianna Mingone\thanks{arianna.mingone@gmail.com}}

\affil[1]{Zeliade Systems, 56 rue Jean-Jacques Rousseau, Paris, France}
\affil[2]{Universit\`{a} degli Studi di Udine, Udine, Italy}

\hypersetup{ pdftitle={No arbitrage SVI} }

\usepackage{amsthm}
\usepackage{cleveref}
\newtheorem{theorem}{Theorem}[section]
\newtheorem{corollary}[theorem]{Corollary}

\newtheorem{lemma}[theorem]{Lemma}
\newtheorem{proposition}[theorem]{Proposition}
\newtheorem{remark}[theorem]{Remark}

\begin{document}
\maketitle
	
\begin{abstract}
	We fully characterize the absence of Butterfly arbitrage in the SVI
formula for implied total variance proposed by Gatheral in 2004. The
main ingredient is an intermediate characterization of the necessary
condition for no arbitrage obtained for any model by Fukasawa in 2012
that the inverse functions of the $-d_1$ and $-d_2$ of the Black-Scholes
formula, viewed as functions of the log-forward moneyness, should be
increasing. A natural rescaling of the SVI parameters and a meticulous
analysis of the Durrleman condition allow then to obtain simple range
conditions on the parameters. This leads to a straightforward
implementation of a least-squares calibration algorithm on the no
arbitrage domain, which yields an excellent fit on the market data we
used for our tests, with the guarantee to yield smiles with no Butterfly
arbitrage.
\end{abstract}

\section{Introduction}\label{introduction}

	Jim Gatheral proposed in 2004 the following \emph{Stochastic Volatility
Inspired} model for the implied total variance (meaning: the square of the
implied volatility times the time-to-maturity):
\[SVI(k) = a+b(\rho(k-m)+\sqrt{(k-m)^2+\sigma^2})\]
where $k$ is the log-forward moneyness, and $(a,b,\rho,m,\sigma)$
parameters.

	This formula quickly became the benchmark at least on Equity markets,
due to its ability to produce very good fits. Fabien Le Floch (head of
research at Calypso) has a blog article on a situation where SVI
\emph{does not fit}, which is a good indicator of how rare such a
situation is in practice. The practitioner literature on SVI and its
variants is plentiful (\cite{damghani2013arbitraging},
\cite{itkin2014one}, \cite{nagy2019volatility},
\cite{bossu2014advanced}, \cite{itkin2020fitting}), and SVI is now part
of every reference textbook on volatility models
(\cite{gatheral2011volatility}, \cite{gulisashvili2012analytically}).

In 2009, the whitepaper on the \emph{Quasi-explicit calibration of
Gatheral's SVI} (\cite{de2009quasi}, also part of Stefano De Marco PHD
thesis) proposed a simple trick to disambiguate the calibration of SVI,
and became itself a reference calibration algorithm.

	SVI has been extended by Gatheral and Jacquier in a seminal paper to
\emph{surfaces} in \cite{gatheral2014arbitrage}, which provides the
first explicit family of implied volatility surfaces with explicit and
tractable no arbitrage conditions, both for the Butterfly and Calendar
Spread arbitrages. SSVI has been extended further in
\cite{guo2016generalized} to other smile shapes, and in
\cite{hendriks2017extended} to correlation parameters functions of the
time-to-maturity. A quick and robust calibration algorithm for the
latter is provided in \cite{corbetta2019robust}.

SSVI smiles (at a fixed time to maturity) are a subset of SVI smiles with 3 parameters instead of 5, and so, for them, an explicit sufficient condition for no (Butterfly) arbitrage is available (cf \cite{gatheral2014arbitrage}). \cite{klassen2016necessary} discusses also partial necessary and sufficient conditions for SSVI smiles.

	A remarkable fact is that, despite the simplicity of the formula, no
Butterfly arbitrage conditions for a SVI smile remained up to now too intricate. So
for instance in the algorithm \cite{de2009quasi} there is no guarantee
that the calibrated parameter will be arbitrage-free. An interesting
practical approach is provided in \cite{ferhati2020robust}, where the no
arbitrage constraints are expressed as a discretized set of Durrleman
conditions and encoded as non-linear constraints in the optimizer;
\emph{stricto sensu} there also, there is no guarantee though that the
calibrated parameter will be arbitrage-free. In this paper, we solve
this long-standing issue.

We start in \cref{the-durrleman-condition-and-no-arbitrage-for-svi} with a precise discussion of the meaning of no
Butterfly arbitrage, which is based on \cite{tehranchi2020black} and on \cite{roper2010arbitrage}
for the corresponding statements in terms of volatility.

We proceed in \cref{gux5f1-and-the-fukasawa-necessary-condition-for-no-butterfly-arbitrage} with a slight generalization of the beautiful result by Fukasawa in
\cite{fukasawa2012normalizing}, which states that the inverse functions of the $-d_1$
and $-d_2$ coefficients of the Black-Scholes formula should be
increasing under no Butterfly arbitrage. We need this generalization
to handle all the configurations of SVI parameters. In this section we also
clarify when and how Call and Put SVI option prices, given by the Black-Scholes formula
with the SVI formula as argument, can be represented as expectations, using 
the results in \cite{tehranchi2020black}.

The main ingredient (\cref{normalizing-svi}) is then to use a natural rescaling of SVI: we work
with the parameters $\alpha,\mu$ where $a=\sigma \alpha$ and
$m = \sigma \mu$, and the dummy variable $l=\frac{k-m}{\sigma}$ instead of
$k$.  It turns out that the Fukasawa conditions for
SVI in the new parameters do not involve $\sigma$. An interesting property of those conditions is that they provide the positivity of the 1st term of the
Durrleman condition; based on the fact that in our case the
complementing 2nd term reads $\frac{1}{2 \sigma} G_2(l)$ where $G_2$
does not depend on $\sigma$, ensuring the Durrleman condition yields a
simple condition on $\sigma$. In \cref{investigating-fukasawa-necessary-no-arbitrage-conditions}
we give the full characterization of the Fukasawa conditions for SVI, and \cref{no-arbitrage-domain-for-svi} finishes the work with the full characterization of no Butterfly arbitrage.

It should be noted that we do not impose the condition $a \geq 0$, as
is often done without justification; we work out the necessary and
sufficient conditions in the full domain of the SVI parameters.

At this stage, we have made fully explicit the domain of the SVI parameters
for which no Butterfly arbitrage holds. It is straightforward to code,
resorting to root finding numerical routines (like the Brent algorithm)
for the evaluation of the thresholds we characterized in our
computations. There are then 2 byproducts of this parametrization of the
domain of high practical interest:
\begin{itemize}
\item
  a quick check routine that a given SVI parameter lies in the domain or
  not, which disentangles between 4 possible situations of arbitrage;
\item
  a calibration algorithm, using any least-squares type objective
  function and a minimizer able to handle bounds.
\end{itemize}

We provide in the last section (\cref{calibration-experiments}) numerical tests performed on
data on index options purchased from the CBOE.

	SVI models a volatility smile, not a volatility surface, so without
ambiguity when we use the no arbitrage wording for SVI, we mean the
absence of Butterfly arbitrage.

	We thank Antoine Jacquier and Stefano De Marco for useful discussions
and comments.

	\subsection{Domain of SVI parameters}\label{domain-of-svi-parameters}
	The SVI model is defined when $a,m\in\mathbb{R}$, $b\geq0$, $|\rho| \leq 1$,
$\sigma \geq 0$. We recall that SVI is a convex function, with a minimum value given by
$a+b\sigma\sqrt{1-\rho^2}$ (possibly attained at infinity if
$|\rho| = 1$) and which goes to infinity as $k$ goes to $\pm\infty$
(for $|\rho|<1$). Since SVI models total variances, it is therefore
required that $a+b\sigma\sqrt{1-\rho^2}\geq0$.

If $\rho=-1$ the SVI smile decreases from $\infty$ to $a$, and if
$\rho=+1$ the SVI smile increases from $a$ to $\infty$.

	\section{The Durrleman condition and no arbitrage for
SVI}\label{the-durrleman-condition-and-no-arbitrage-for-svi}

	Let $S_0$ denote the underlying asset value of standard Call options
with a fixed maturity $t>0$. Without loss of generality we assume that
there is no interest rates nor dividend rates. In case of deterministic
interest $r$ and dividend rates $\delta$, all the statements in this
section still hold once $S_0$ is replaced by the Forward corresponding
to the option maturity $F_t = S_0 \exp{ \int_0^t (r_s-\delta_s) ds}$ and
working with the num\'{e}raire of the option maturity.

	\subsection{Axiom of no Butterfly
arbitrage}\label{axiom-of-no-butterfly-arbitrage}

	The condition of no Butterfly arbitrage is achieved when the Call price
function with respect to the strike is (we follow the very careful
treatment in \cite{tehranchi2020black}):
\begin{enumerate}
\def\labelenumi{\arabic{enumi}.}
\itemsep1pt\parskip0pt\parsep0pt
\item
  convex;
\item
  non-increasing;
\item
  with value in the range $[(S_0-K)^+, S_0]$.
\end{enumerate}

These properties assume only that there is a perfect market for the
underlying asset and for the Call options, with short-selling allowed, and
that there is no static buy-sell strategy involving the underlying asset and a
finite set of Call options with a Profit and Loss which is strictly
positive.

We recall in particular that the \emph{large moneyness behaviour}
stating that the Call price function should go to zero at $\infty$ is
an additional assumption, and does not strictly follow from the
no arbitrage axiom.

In the case of a Call price function specified through an implied
volatility: $C(K) = C_{BS}(k, \sqrt{w(k)})$ where $w$ is the implied total
variance $\sigma^2 t$ and $C_{BS}(k,a)$ is the Black-Scholes formula
expressed as a function of the log-forward moneyness
$k=\log{\frac{K}{S_0}}$ and the implied total volatility, the 3rd property
is automatically granted since the $C_{BS}$ function is increasing with
respect to its 2nd argument and since the range bounds correspond
to the limit when $a$ goes to $0$ and $\infty$.

Observe now that if the 3rd property is satisfied, then the 1st one
implies the 2nd one because an increasing convex function cannot be
bounded.

	\subsection{Smiles vanishing at some
point}\label{smiles-vanishing-at-some-point}

	Can a volatility smile reach $0$ at some (finite) point? Assume that it
is the case, so $w(k_m)=0$ at the log-forward moneyness $k_m$
corresponding to some strike $K_m$. Then it means that the Call price
with this strike is equal to its intrinsic value $(S_0-K_m)_+$. If $K_m$
lies on the right of $S_0$, the price is therefore $0$, and by the
property 2 above all the Call prices with $K>K_m$ will also be $0$. If
$K_m$ lies on the left of $S_0$, the option price is $S_0-K_m$; as
the option price with a strike $0$ is equal to $S_0=S_0-0$, the
convexity property implies that all the Call prices with $K< K_m$ are
smaller than $S_0-K$ which is the value of the chord between the
points $0$ and $K_m$. Since this value $S_0-K$ is also lower bound for
the Call prices, they are eventually equal to this value. So, in the
implied volatility space, this means that $w=0$ for $K \geq K_m$ in the
1st case, and $w=0$ for $K \leq K_m$ in the second case.

This means that no arbitrage implies that smiles reaching $0$ above
(respectively below) the At-The-Money (forward) point will vanish above
(respectively below) this point. In the case of SVI, smiles reach zero
at most at a single strike, and only if $a+b\sigma\sqrt{1-\rho^2}=0$ and
$|\rho|<1$, in which case they are strictly positive for
other strike values, and there is a Butterfly arbitrage. So we can
discard this case and assume $a+b\sigma\sqrt{1-\rho^2}>0$ when
$|\rho|<1$.

	\subsection{No Butterfly arbitrage criterion for
SVI}\label{no-butterfly-arbitrage-criterion-for-svi}

	At this stage we know that SVI smiles with no Butterfly arbitrage are
positive, and that the 3rd property above is automatically satisfied. So
there is no Butterfly arbitrage if and only if the 1st property holds.
Now for positive smiles, as recalled in \cite{gatheral2014arbitrage}
after Lemma 2.2, with $w(k)=SVI(k)$:
\begin{equation}\label{eqpK}
	\begin{split}
	p(K):={\frac{\partial^2 C_{BS}}{\partial K^2}}\Biggr|_{K=S_0 e^k} &= {\frac{\partial^2 C_{BS}(k, \sqrt{w(k)})}{\partial K^2}}\Biggr|_{K=S_0 e^k}\\
	&=\frac{g(k)}{S_0e^k\sqrt{2 \pi w(k)}}\exp{\biggl(-\frac{d_2(k,\sqrt{w(k)})^2}{2}\biggr)}
	\end{split}
\end{equation}
where $d_2$ is the standard coefficient of the Black-Scholes formula:
\[d_{1,2}(k,\sigma) = -\frac{k}{\sigma} \pm \frac{\sigma}{2}.\]
So convexity is equivalent to ask the function $g(k)$
(\cite{gatheral2014arbitrage}, equation 2.1)
\begin{equation}\label{eqgSVI}
g(k) := \biggl(1-\frac{kSVI'(k)}{2SVI(k)}\biggr)^2 - \frac{SVI'(k)^2}{4}\biggl(\frac{1}{SVI(k)}+\frac{1}{4}\biggr) + \frac{SVI''(k)}{2}
\end{equation}
to be non-negative, which is usually called the \emph{Durrleman condition}
(cf.~Theorem 2.9, condition $(IV3)$ of \cite{roper2010arbitrage}).

	Note that the first derivative of the Call function with respect to the
strike necessarily goes to zero as $K$ goes to $\infty$, and to a finite
limit between $-1$ and $0$ as $K$ goes to $0$, which means that the
total mass of $p$ is less than one, but not necessarily one, meaning
that there could be a non-zero mass at zero. It will sum to one if and
only if the limit is $-1$; in this case, $p$ can be interpreted as a
probability measure; the expectation of the underlying asset under this
measure will be strictly less than the underlying asset value, unless the
additional property that the Call price vanishes at infinity holds, in
which case it will be exactly the underlying asset value (cf.~Theorem 2.1.2 of
\cite{tehranchi2020black}).

	The above discussion can be translated in properties of the smile: we
know from Theorem 2.9 in \cite{roper2010arbitrage} that the large
moneyness behaviour is one-to-one with the fact that $d_1(k)$ goes to
$-\infty$ at infinity:
\[\lim_{k \to \infty} d_1(k,\sqrt{w(k)})=-\infty.\]

The fact that there is no mass at zero, or, equivalently, that the
derivative of the Call price with respect to the strike goes to $-1$
when the strike goes to zero, is equivalent to (cf.
\cite{fukasawa2012normalization}, Proposition 2.4):
\[\lim_{k \to - \infty} d_2(k,\sqrt{w(k)})=+\infty.\]

In the case of SVI, the 1st condition translates to $b(1+\rho)< 2$ and
the second one to $b(1-\rho)< 2$. In particular the following Lemma holds:

\begin{lemma}
	In SVI, the limit of $d_1(k,\sqrt{w(k)})$ for $k$ going to $\infty$, is
	\begin{itemize}
		\item $-\infty$ if $b(1+\rho)<2$;
		\item $0$ if $b(1+\rho)=2$;
		\item $\infty$ if $b(1+\rho)>2$.
	\end{itemize}	
\end{lemma}

The proof is simple and it is omitted. An important consequence to this result is that when $b(1+\rho)=2$, the Call prices do not go to zero when the strike goes to infinity and so they are not given as the expectation of the payoff; we will come back to this situation in detail in \cref{gux5f1-and-the-fukasawa-necessary-condition-for-no-butterfly-arbitrage}. In such a case, this does not necessarily lead to an arbitrage and so the request $b(1+\rho)<2$ is not a necessary condition fo the absence of arbitrage. We can summarize the previous discussion as follows:

\begin{proposition}[No Butterfly arbitrage criterion for SVI]
A necessary condition for no Butterfly arbitrage to hold in SVI is that $SVI(k)>0$ for all $k$. Under this condition, there is no arbitrage in SVI if and only if the function $g$ in \cref{eqgSVI} is non-negative. In this case, the function $p(K)$ in \cref{eqpK} where $K=S_0 e^k$, and $S_0$ is the underlying asset value, defines a positive density on $\mathbb R_+$ such that $\int p(x) dx\leq 1$.

Moreover, the Call prices in SVI go to zero when the strike goes to infinity if and only if $b(1+\rho)< 2$, and the derivative of the Call price (expressed in num\'{e}raire of the maturity) with respect to the strike goes to $-1$ if and only if $b(1-\rho)<2$. In the first case $\int x p(x) dx=S_0$ and in the second case $\int p(x) dx=1$.

\end{proposition}

	Note that the two conditions $b(1+\rho) = 2$ and $b(1-\rho) = 2$ can occur simultaneously if and only if $b=2$ and $\rho=0$. We turn now to 
	the weak no Butterfly condition obtained by Fukasawa. Characterizing this intermediary condition will eventually lead us to our full characterization result.

\section{Fukasawa necessary condition for no Butterfly
arbitrage}\label{gux5f1-and-the-fukasawa-necessary-condition-for-no-butterfly-arbitrage}
We recall the beautiful model-free necessary no Butterfly arbitrage condition obtained by Fukasawa in \cite{fukasawa2012normalizing}. Following Fukasawa, let us denote the Black-Scholes prices as $C_{BS}(k,\sigma) = S_0\Phi(d_1(k,\sigma)) - S_0e^k\Phi(d_2(k,\sigma))$ for Calls and $P_{BS}(k,\sigma) = S_0e^k\Phi(-d_2(k,\sigma)) - S_0\Phi(-d_1(k,\sigma))$ for Puts; the implied total volatility is $\sqrt{w(k)}=\sigma(k)$; for a given total implied volatility let us set 
	$$f_{1,2}(k) = -d_{1,2}(k,\sigma(k)).$$
	Fukasawa proved in Theorem 2.8 of \cite{fukasawa2012normalizing}
that (under the hypothesis that option prices are given by the expectation of their payoff) if a total variance smile $w$, expressed as a function of the
log-forward moneyness, has no Butterfly arbitrage, then the two
functions $f_1$ and $f_2$ are necessarily strictly increasing with $f'_{1,2} > 0$.

\subsection{A slight generalization of Fukasawa result}
	
In the following, we generalize Fukasawa's result to the case where the only request on the Put prices is their convexity and differentiability, without requiring that they are given by the expectation of the payoff. 
Note that the proof is essentially Fukasawa's one. This will allow us to cover the boundary case $b(1+\rho) = 2$.

\begin{lemma}
	
	Let Put prices be defined as the Black-Scholes Put prices with a total volatility $\sigma(k)$: $P(K) = P_{BS}(k,\sigma(k))$, where $K=S_0e^k$. If the function $P$ is convex and the total volatility is differentiable, then the functions $f_{1,2}$ are increasing.
	
\end{lemma}
\begin{proof}
	
	Because the total volatility is differentiable, then also the Put prices are differentiable. Define $D_{BS}(K) := \frac{1}{K}\frac{\partial P_{BS}}{\partial k}(k,\sigma)\mid_{\sigma=\sigma(k)} = \Phi(f_2(k))$ and $D(K) := \frac{d\,P}{d\,K}(K)$. Note that here we do not use the equality $D(K) = E[I_{K>S_T}]$ whose proof requires that the Put prices are the expectation of their payoff. Since the Put prices are given by the  Black-Scholes formula, they lie between $(K-S_0)^+$ and $K$, and since the function $K \to P(K)$ is convex, then its derivative lies necessarily between $0$ and $1$: $0\leq D(K)\leq 1$. It holds
	\begin{equation}\label{eqDFuk}
		\begin{aligned}
			D(K) &= \frac{d}{dK}P_{BS}\bigl(\log(K/S_0),\sigma(\log(K/S_0))\bigr)\\
			&= D_{BS}(K) + \frac{1}{K}\frac{\partial P_{BS}}{\partial \sigma}\bigl(\log(K/S_0),\sigma(\log(K/S_0))\bigr)\frac{d\sigma}{dk}\bigl(\log(K/S_0)\bigr)\\
			&= D_{BS}(K) + \phi\bigl(f_2\bigl(\log(K/S_0)\bigr)\bigr)\frac{d\sigma}{dk}\bigl(\log(K/S_0)\bigr).
		\end{aligned}
	\end{equation}
	
	We now check that
	\begin{equation*}
		f_2(k)\frac{d\sigma}{dk}(k)<1.
	\end{equation*}
	From the previous equations, $\frac{d\sigma}{dk}(k) = \frac{D(S_0e^k)-D_{BS}(S_0e^k)}{\phi(f_2(k))}$ and because of the bounds for $D(K)$, this quantity lies in $\Bigl[-\frac{1-\Phi(-f_2(k))}{\phi(-f_2(k))},\frac{1-\Phi(f_2(k))}{\phi(f_2(k))}\Bigr]$. So when $f_2(k)$ is non-negative, $f_2(k)\frac{d\sigma}{dk}(k)\leq f_2(k)\frac{1-\Phi(f_2(k))}{\phi(f_2(k))}$. Otherwise, $f_2(k)\frac{d\sigma}{dk}(k)\leq -f_2(k)\frac{1-\Phi(-f_2(k))}{\phi(-f_2(k))}$. Both quantities are less than $1$.
	
	At this point, we verify that
	\begin{equation}\label{eqf1Fuk}
		f_1(k)\frac{d\sigma}{dk}(k)<1.
	\end{equation}
	It holds $KD(K) \geq P(K)$. Indeed, since $P$ is convex, its tangent at $K$ lies below the function $P$ itself, so for any $x \geq 0$ one has $P(K)+(x-K)D(K) \leq P(x)$ and evaluating at $x=0$ we obtain the target inequality since $P(0)=0$. From this inequality, using \cref{eqDFuk} and writing the explicit formula of $P(K)$, one gets $0\leq S_0\Phi(f_1(k))+S_0e^k\phi(f_2(k))\frac{d\sigma}{dk}(k) = S_0\Phi(f_1(k))+S_0\phi(f_1(k))\frac{d\sigma}{dk}(k)$. If $f_1$ is non-positive, then $f_1(k)\frac{d\sigma}{dk}(k)\leq -f_1(k)\frac{1-\Phi(-f_1(k))}{\phi(-f_1(k))}<1$. Otherwise, \cref{eqf1Fuk} is automatically satisfied when $\frac{d\sigma}{dk}(k)$ is non-positive, while when it is positive, we notice that $f_1(k)\frac{d\sigma}{dk}(k) = f_2\frac{d\sigma}{dk}(k) - \sigma\frac{d\sigma}{dk}(k) < 1-\sigma\frac{d\sigma}{dk}(k) <1$.
	
	Finally, we show that $f_1$ and $f_2$ are increasing. Indeed, from their definition, $\frac{df_{1,2}}{dk}(k) = \frac{1}{\sigma(k)}\bigl(1-\frac{d\sigma}{dk}(k)\bigl(\frac{k}{\sigma(k)}\pm\frac{\sigma(k)}2\bigr)\bigr) = \frac{1}{\sigma(k)}\bigl(1-\frac{d\sigma}{dk}(k)f_{2,1}(k)\bigr) >0$.
\end{proof}

What if we start from convex Call prices defined by the Black-Scholes Call prices instead of Put prices? In such case, it would be enough to prove that also the Put prices come from the Black-Scholes Put prices and they are convex. Indeed, synthetizing the strategy of selling a Call with strike $K$, buying the underlying and selling a quantity $K$ of cash at time $0$ yields a payoff $(X_T-K)^+ - X_T+K =(K-X_T)^+$, where $X_T$ is the realized value of the underlying at maturity, so that the assumption of no arbitrage leads to $P(K)=C(K)-S_0+K$ which is the Put-Call parity. Since the equality $P_{BS}(k,\sigma(k))=C_{BS}(k,\sigma(k))-S_0+K$ also holds from the definition of the functions $C_{BS}$ and $P_{BS}$, it follows firstly that $P(K)=P_{BS}(k,\sigma(k))$, so the Calls and Puts with the same strike have the same implied volatility. Secondly, looking at the Put-Call parity, one notices that the Call prices are convex iff the Put prices are convex. Applying the previous Lemma one finds again $f_{1,2}'(k)>0$.

\subsection{Expectation-based representation of the Calls and Put prices in SVI}

The issue of the representation of the Call and Put price functions by an expectation under our purely analytical assumptions has been settled by Tehranchi in \cite{tehranchi2020black}. Indeed re-starting from the assumption that the implied volatility function is such that the Call price function $K \to C(K)$ is convex, it follows from the above discussion that we are exactly in the situation of Theorem 2.1.2 in \cite{tehranchi2020black}. So there exists a non-negative random variable $S_T$ such that $E[S_T] \leq S_0$ and $C(K)=S_0-E[K \wedge S_T]$. We have then from the above discussion that $P(K)=C(K)+K-S_0=K-E[K \wedge S_T]$. It is interesting to note that:
$$C(K) = S_0-E[K \wedge S_T] = E[(S_T-K)^+]+S_0-E[S_T],$$
whereas the usual expectation for the Put formula still holds:
$$P(K) = K-E[K \wedge S_T] = E[(K-S_T)^+].$$

Going back to SVI, when $b(1+\rho)<2$ we are in the situation where $C(K) \to 0$ when $K \to \infty$, so that $E[S_T]=S_0$ in the above representation. In the case $b(1+\rho)=2$ one has $\lim_{k \to \infty} f_1(k)=0$, which plugged into the Black-Scholes formula gives $\lim_{K \to \infty} C(K) = \frac{S_0}{2}$; in turns this gives $E[S_T]=\frac{S_0}{2}$. We can state the following:

\begin{proposition}

Let $C(K):=C_{BS}(k,\sqrt{SVI(k)})$ and $P(K):=P_{BS}(k,\sqrt{SVI(k)})$ be the Call and Put prices in SVI. Then
there exists a positive random variable $S_T$ such that:
\begin{enumerate}
\item $P(K) = K-E[K \wedge S_T] = E[(K-S_T)^+]$ and $C(K)= S_0-E[K \wedge S_T] = E[(S_T-K)^+]+S_0-E[S_T]$;
\item if $b(1+\rho)<2$, $C(K) \to 0$ when $K \to \infty$, $E[S_T]=S_0$ and $C(K)= E[(S_T-K)^+]$;
\item if $b(1+\rho)=2$, $C(K) \to \frac{S_0}{2}$ when $K \to \infty$, $E[S_T]=\frac{S_0}{2}$ and  $C(K)= E[(S_T-K)^+]+\frac{S_0}{2}$.
\end{enumerate}

\end{proposition}

The last ingredient we will require is a natural change of parameters in SVI, that we describe in the next section altogether with the main argument of our full characterization.

\section{Normalizing SVI}\label{normalizing-svi}

	We now rescale SVI in the following way, which is natural:
\begin{align*}
SVI(k) &= \alpha\sigma + b\sigma\biggl(\rho\frac{k-m}{\sigma} + \sqrt{\Bigl(\frac{k-m}{\sigma}\Bigr)^2 + 1}\biggr)\\
 &= \sigma N\biggl(\frac{k-m}{\sigma}\biggr)
\end{align*}
with $\alpha := a/\sigma$ and $N(l) := \alpha+b(\rho l+\sqrt{l^2+1})$. With this
rewriting, the derivatives of the SVI model become
\begin{align*}
SVI'(k) &= N'\biggl(\frac{k-m}{\sigma}\biggr),\\
SVI''(k) & = \frac{1}{\sigma}N''\biggl(\frac{k-m}{\sigma}\biggr).
\end{align*}

	Observe that the second derivative $N''$ is positive so $N$ is
strictly convex. Its only critical point is a minimum that we call
$l^* = -\frac{\rho}{\sqrt{1-\rho^2}}$. We gather the important
properties of $N$ in the following:

	\begin{lemma}[Normalized SVI]
Let $N(l) := \alpha+b(\rho l+\sqrt{l^2+1})$ where $a = \alpha \sigma$. Then $N$ is strictly convex with a minimum at $l^* = -\frac{\rho}{\sqrt{1-\rho^2}}$, where $N(l^*)=\alpha+b \sqrt{1-\rho^2}$. Also:
\begin{align*}
N'(l) &= b\biggl(\rho +\frac{l}{\sqrt{l^2+1}}\biggr),\\
N''(l) &= \frac{b}{(l^2+1)^{\frac{3}{2}}}.
\end{align*}
In particular as $l\to\pm\infty$:
\begin{align*}
N(l) \sim \alpha + b(\rho\pm 1)l, & & N'(l) \to b(\rho\pm 1), & &N''(l) \to 0,
\end{align*}
and for every $k$  
$$SVI(k)=\sigma N\Bigl(\frac{k-m}{\sigma}\Bigr).$$

\end{lemma}

	In the above Lemma, note that the statement $N(l) \sim a' + b(\rho\pm 1)l$ covers the cases $b=0$ and $b \neq 0$.
	
	Hereafter we will also put $m=\mu\sigma$, so that $k=\sigma (l + \mu)$
and
\[SVI_{a,b,\rho, m,\sigma}(k) = \sigma N_{\alpha,b,\rho}\biggl(\frac{k}{\sigma}-\mu\biggr)\]
	where the parameters have the following constraints:
\begin{align*}
&b \geq 0,
&&|\rho| \leq 1,
&&\mu\in\mathbb{R},
&&\sigma \geq 0,
&&\alpha+b\sqrt{1-\rho^2} \geq 0.
\end{align*}

\subsection{Expressing $g$ with $f_{1,2}$ in rescaled parameters, and our main argument} \label{mainArg}
There is a nice expression of
$g$ involving the functions $f_{1,2}$; indeed as shown e.g.~in
\cite{de2018moment} (Eq. 55 p.~25):
\[{\frac{\partial^2 C}{\partial K^2}}\Biggr|_{K=S_0 e^k}=\phi(f_2(k)) \left(f_1'(k)f_2'(k)\sqrt{w(k)}  + (\sqrt{w})''(k) \right) \frac 1{S_0 e^k}\]
where $\phi$ is the standard Gaussian density. By identification this
yields
\[ g(k) = \Bigl(f_1'(k)f_2'(k)\sqrt{w(k)}  +( \sqrt{w})''(k) \Bigr)\sqrt{ w(k)}.\]

With our rescaled parameters, we have
\begin{equation*}
g(k) = \biggl(1-\frac{kN'\bigl(\frac{k}{\sigma}-\mu\bigr)}{2\sigma N\bigl(\frac{k}{\sigma}-\mu\bigr)}\biggr)^2 - \frac{N'\bigl(\frac{k}{\sigma}-\mu\bigr)^2}{4}\biggl(\frac{1}{\sigma N\bigl(\frac{k}{\sigma}-\mu\bigr)}+\frac{1}{4}\biggr) + \frac{N''\bigl(\frac{k}{\sigma}-\mu\bigr)}{2\sigma}
\end{equation*}
and writing $G(l):=g(\sigma (l + \mu))$ we find
\[G(l)=\biggl(1-\frac{(l+\mu)N'(l)}{2N(l)}\biggr)^2 - \frac{N'(l)^2}{4}\biggl(\frac{1}{\sigma N(l)}+\frac{1}{4}\biggr) + \frac{N''(l)}{2\sigma}.\]

We can rewrite $G$ as
\begin{align*}
G(l) &= \biggl(1-N'(l) \biggl(\frac{(l+\mu)}{2N(l)}+\frac{1}{4} \biggr)\biggr) \biggl(1-N'(l) \biggl(\frac{(l+\mu)}{2N(l)}-\frac{1}{4} \biggr)\biggr) + \frac{1}{2\sigma}\biggl(N''(l)-\frac{N'(l)^2}{2N}\biggr)\\ &= G_1(l) + \frac{1}{2\sigma} G_2(l)
\end{align*}
where
\begin{align*}
G_1(l) &:= \biggl(1-N'(l) \biggl(\frac{(l+\mu)}{2N(l)}+\frac{1}{4} \biggr)\biggr) \biggl(1-N'(l) \biggl(\frac{(l+\mu)}{2N(l)}-\frac{1}{4} \biggr)\biggr),\\
G_2(l) &:= N''(l)-\frac{N'(l)^2}{2N(l)}.
\end{align*}

Call $G_{1+}$ the first factor of $G_1$ and $G_{1-}$ the second one. Then
$f'_{1,2}(\sigma(l+\mu)) = f'_{1,2}(k) = \frac{1}{\sqrt{\sigma N\bigl(\frac{k}{\sigma}-\mu\bigr)}}\Bigl(1 - N'\bigl(\frac{k}{\sigma}-\mu\bigr)\Bigl(\frac{k}{2\sigma N\bigl(\frac{k}{\sigma}-\mu\bigr)}\pm\frac{1}{4}\Bigr)\Bigr)= \frac{G_{1\pm}(l)}{\sqrt{\sigma N(l)}}$ and the
Fukasawa conditions correspond to $G_{1\pm} > 0$, which entails that
$G_1 > 0$.

Completing the identification yields
$G_1\bigl(\frac{k}{\sigma}-\mu\bigr) = f'_{1}(k)f'_{2}(k)w(k)$ and
$\frac{1}{2\sigma}G_2\bigl(\frac{k}{\sigma}-\mu\bigr) = (\sqrt{w})''(k)\sqrt{w(k)}$.

It is now instrumental to observe that:
$$g(k)=G(l)=G_1(l) + \frac{1}{2\sigma} G_2(l)$$
where
\begin{itemize}
\itemsep1pt\parskip0pt\parsep0pt
\item
  $G_1$ depends only on $\alpha, b, \rho, \mu$,
\item
  $G_2$ depends only on $\alpha, b, \rho$,
\end{itemize}
so that the dependency of $G$ in $\sigma$ is particularly simple; this
is the main benefit of our rescaling of SVI.

Our main argument is now as follows: the Fukasawa conditions yield that is it necessary that $G_1>0$; for a given choice of $b, \rho, \alpha$, this will
give a condition on $\mu$, which therefore characterizes the Fukasawa conditions in SVI. Given then a parameter $\mu$ satisfying this condition, the positivity of $g$ (or $G$) can be casted as a simple condition on $\sigma$:
$$\sigma \geq \sup_l \frac{G_2(l)}{G_1(l)}$$
which yields the full characterization of no Butterfly arbitrage in SVI. 

In \cref{investigating-fukasawa-necessary-no-arbitrage-conditions} we investigate the conditions on $G_1$ related to the Fukasawa conditions, and in \cref{no-arbitrage-domain-for-svi} this latter condition on $\sigma$.

\subsection{Classifying the normalized SVI parameters}

We will use the following notations to clarify the assumptions made on the SVI parameters:
\begin{itemize}
\item (A1) $\alpha+b \sqrt{1-\rho^2}>0$ and $|\rho|<1$,
\item (A2) $\alpha \geq 0$ and $|\rho|=1$,
\item (B1) $b(1\pm \rho)<2$,
\item (B2)  $b(1+\rho)<2$ and $b(1-\rho)=2$,
\item (B3)  $b(1+\rho)=2$ and $b(1-\rho)<2$,
\item (B4)  $b(1+\rho)=2$ and $b(1-\rho)=2$, which is equivalent to $b=2, \rho=0$.
\end{itemize}

In the sequel, to avoid singularities in our computations, we will assume $b$ positive since the case $b=0$ is the Black-Scholes case, which is a trivial case of no arbitrage, and exclude the boundary cases $|\rho| = 1$, so work under assumption (A1). We revisit those boundary cases in \cref{the-monotonous-case-rho-pm-1} where we will assume (A2).

	\section{Investigating Fukasawa necessary no arbitrage
conditions}\label{investigating-fukasawa-necessary-no-arbitrage-conditions}

	\subsection{Limits at infinity}\label{limits-at-infinity}

	We have the following:

	\begin{lemma}[Limits of $G_1$]
$$\lim_{\pm \infty} G_1(l) = \Bigl(\frac{1}{2}-\frac{b(\rho\pm 1)}{4}\Bigr)\Bigl(\frac{1}{2}+\frac{b(\rho\pm 1)}{4}\Bigr).$$

In particular, $G_1(\infty)\geq0$ and $G_1(-\infty)\geq0$ iff simultaneously $b(1\pm\rho) \leq 2$.

\end{lemma}

	These conditions are conditions on the asymptotic slopes of the total
variance smile, and are therefore related to the Roger Lee Moment
formula \cite{lee2004moment}; this is a general fact for the Fukasawa
conditions: \cite{fukasawa2012normalizing} contains several asymptotic
statement on $f_1$ and $f_2$ which are directly related to the
asymptotic behaviour of $\frac{w(k)}{k}$.

	\subsection[The conditions as an interval for mu]{The conditions as an interval for
$\mu$}\label{the-conditions-as-an-interval-for-n}

	Let us investigate the corresponding Fukasawa conditions of positivity
of $G_{1+}$ and $G_{1-}$ in terms of SVI parameters. We start with the following:

\begin{lemma}

Let \begin{equation}\label{eqLmp}
L_\pm(l;\alpha,b,\rho):=2N(l)\Bigl(\frac{1}{N'(l)}\mp\frac{1}{4}\Bigr)-l
\end{equation}
where $L_+$ is defined on $]l^*,+\infty[$ and $L_-$ on $]-\infty,l^*[$.Then $G_{1\pm}>0$ if and only if $\sup_{l < l^*}L_-(l) <  \inf_{l > l^*}L_+(l)$ and
$$\mu \in I_{\alpha,b,\rho} := ]\sup_{l < l^*}L_-(l), \inf_{l > l^*}L_+(l)[.$$

\end{lemma}

	\begin{proof}
In order to have $G_{1\pm}>0$, we need $\sup_{l < l^*}L_\pm(l) < \mu < \inf_{l > l^*}L_\pm(l)$. Indeed $G_{1\pm}(l) = 1-N'(l) \Bigl(\frac{(l+\mu)}{2N(l)}\pm\frac{1}{4} \Bigr)$ so that $G_{1\pm}(l)>0$ iff $1\mp \frac{N'(l)}{4} > N'(l) \frac{(l+\mu)}{2N(l)}$. Since $L_+(l)< L_-(l)$ for every $l$, we obtain an interval for $\mu$ given by $\sup_{l < l^*}L_-(l) < \mu < \inf_{l > l^*}L_+(l).$
\end{proof}

	\begin{remark} In order to alleviate the notations, we will often suppress the list of parameters in $L_\pm$, or when we need it just denote the dependency in $\alpha$, $(b,\rho)$ being fixed.

\end{remark}

	What are the basic properties of $L_-$ and $L_+$?

Note that $L_-(l^{*-}) = -\infty$ and, under $b(1-\rho)<2$,
$L_-(-\infty)=-\infty$. It follows that $l_-$ such that
$L_-(l_-) = \sup_{l < l^*}L_-(l)$ lays in $]-\infty,l^*[$. Similarly,
$L_+(l^{*+}) = +\infty$ and $L_+(+\infty)=+\infty$ when $b(1+\rho)< 2$,
so $l_+$ such that $L_+(l_+) = \inf_{l > l^*}L_+(l)$ lays in
$]l^*,+\infty[$. When $b(1-\rho)=2$ then $L_-(-\infty)=-\frac{\alpha}{2}$
while when $b(1+\rho)=2$ then $L_+(+\infty)=\frac{\alpha}{2}$. Indeed at
infinity $L_-$ behaves as
$2\alpha\Bigl(\frac{1}{b(\rho-1)}+\frac{1}{4}\Bigr) + \frac{2+b(\rho-1)}{2}l$
while $L_+$ as
$2\alpha\Bigl(\frac{1}{b(\rho+1)}-\frac{1}{4}\Bigr) + \frac{2-b(\rho+1)}{2}l$.
In these cases the supremum of $L_-$ (or the infimum of $L_+$), could be
reached at $-\infty$ (or $+\infty$).

	Experiments show that not every choice of $(\alpha,b,\rho)$ leads to
$L_-(l)< -\epsilon< 0$ for all $l< l^*$ and $L_+(l)> \epsilon> 0$ for
all $l> l^*$, so the interval for $\mu$ could be empty: for example, for
$\alpha=-0.8, b=1$ and $\rho=0.5$, we have $L_-(l_-)>L_+(l_+)$. This
suggests that the situation is intricate; we show below that when
$\alpha\geq0$, the interval is non-empty.

	\subsubsection[The case alpha non-negative]{The case $\alpha\geq 0$}\label{the-case-ageq-0}

	In the case $\alpha \geq 0$, we can indeed demonstrate that the interval for
$\mu$ is non-empty, with the following easy argument:

$L_-$ is negative for $l< l^*$ iff $\frac{N}{2N'}(4+N')-l$ is negative.
In this domain $N'$ is negative, so the previous condition is equivalent
to ask $N(4+N') - 2lN' > 0$, or equivalently $2(N-lN') +N(2 + N') > 0$.
Let us consider the first term. We have
$N-lN' = \alpha + b\sqrt{l^2+1} - \frac{bl^2}{\sqrt{l^2+1}}$ which is
greater than $0$ iff, multiplying by $\sqrt{l^2+1}$, also
$\alpha\sqrt{l^2+1}+b>0$ or equivalently $\alpha>-\frac{b}{\sqrt{l^2+1}}$. This
holds for $\alpha\geq 0$ (note that the latter quantity reaches its maximum
at $-\infty$ where it equals $0$, so this proof cannot handle the case
$\alpha < 0$).

We can now consider the second term. We want $2 + N'> 0$. Since
$N'>b(\rho-1)$, then $2 + N'> 2+b(\rho-1) \geq 0$. So $L_-$ is always
strictly negative for $l< l^*$ and $\alpha\geq 0$.

Similarly, $L_+$ is positive for $l> l^*$ iff $2(N-lN') +N(2 - N') > 0$.
With the same arguments as before we obtain that $L_+$ is strictly
positive for $l> l^*$ and $\alpha\geq 0$.

	Under (B1), we showed $L_-(-\infty)=-\infty$ and $L_+(\infty)=\infty$, so the
interval $I$ is non-empty.

When $b(1-\rho)=2$ or $b(1+\rho)=2$ this result is still valid. Since
in such cases $L_-(-\infty)=-\frac{\alpha}{2}$ and
$L_+(+\infty)=\frac{\alpha}{2}$ respectively, then $L_-$ is negative in
$[-\infty,l^*[$ while $L_+$ is positive in $]l^*,+\infty]$ for $\alpha>0$.
Otherwise if $\alpha=0$, $\sup_{l < l^*}L_-(l) = L_-(-\infty) = 0$ and
$\inf_{l > l^*}L_+(l)=L_+(+\infty)=0$ respectively.

	We have proven the following:

	\begin{lemma}[SVI parameters fulfilling Fukasawa necessary no arbitrage conditions: case $\alpha \geq 0$]\label{LemmaaPositive}
Assume (A1). For every $(\alpha,b,\rho)$ with $\alpha\geq 0$:
\begin{itemize}
\item under (B1), the interval $I_{\alpha,b,\rho}$ is non-empty and contains $0$;
\item under (B2),
\begin{itemize}
\item if $\alpha> 0$, the interval $I_{\alpha,b,\rho}$ is non-empty and contains $0$;
\item if $\alpha=0$, the interval $I_{0,b,\rho}$ is non-empty and has $0$ as left boundary;
\end{itemize}
\item under (B3),
\begin{itemize}
\item if $\alpha> 0$, the interval $I_{\alpha,b,\rho}$ is non-empty and contains $0$;
\item if $\alpha=0$, the interval $I_{0,b,\rho}$ is non-empty and has $0$ as right boundary;
\end{itemize}
\item under (B4),
\begin{itemize}
\item if $\alpha> 0$, the interval $I_{\alpha,2,0}$ is non-empty and contains $0$;
\item if $\alpha=0$, the interval $I_{0,2,0}$ is empty.
\end{itemize}

\end{itemize}

\end{lemma}

	\subsubsection[Computation of the interval for
	mu under (B1)]{Computation of the interval for
$\mu$ under (B1)}\label{computation-of-the-interval-for-n}

	We tackle now the computation of the interval for $\mu$ in the general
case where $\alpha$ is not necessarily positive, which is less
straightforward. In this section we will assume (B1); we
deal with the other cases in the dedicated \cref{the-case-b1-rho2-or-b1rho2}.

	We consider the function $L_-$ for $l< l^*$ and $L_+$ for $l> l^*$. We
have $L'_\pm(l) = 1\mp\frac{N'}{2}-\frac{2NN''}{N'^2}$ and it follows
that $L'_-(l_-)=L'_+(l_+)=0$.

	The corresponding equations in $l$ are:
\begin{equation*}
1\mp\frac{b}{2}\biggl(\rho +\frac{l}{\sqrt{l^2+1}}\biggr)-\frac{2(\alpha+b(\rho l+\sqrt{l^2+1}))}{b\sqrt{l^2+1}(\rho\sqrt{l^2+1} + l)^2} = 0.
\end{equation*}
Actually, we don't need to solve these equations. Accordingly, we set:
\begin{equation}\label{eqgmexplicit}
g_{\pm(b, \rho)}(l) = \Bigl(\rho\sqrt{l^2+1} + l\Bigr)^2\biggl(\sqrt{l^2+1}\biggl(\frac{1}{2}\mp\frac{b \rho}{4}\biggr) \mp \frac{bl}{4} \biggl)-\Bigl(\rho l+\sqrt{l^2+1}\Bigr)
\end{equation}
where $g_{+(b, \rho)}$ is defined on $[l^*, \infty[$ and $g_{-(b, \rho)}$ on $]-\infty, l^*]$. Then
$L'_\pm(l)=0$ iff $g_{\pm(b, \rho)}(l)=\frac{\alpha}{b}$.

	The following technical result turns to be a key one:

	\begin{proposition}\label{PropLmgm}

Assume (A1) and (B1), and let $g_{\pm(b, \rho)}$ defined by \cref{eqgmexplicit}. Then $g_{\pm(b, \rho)}(l^*)=-\sqrt{1-\rho^2}$, $g_{\pm(b, \rho)}(\pm \infty)=\infty$, and $g_{\pm(b, \rho)}$ is either monotonous or with a single minimum. Let $s_\pm=l^*$ if $g_{\pm(b, \rho)}$ is monotonous and $s_\pm \neq l^*$ such that $g_{\pm(b, \rho)}(s_\pm)=-\sqrt{1-\rho^2}$ otherwise. Then:
\begin{itemize}
\item $L_-(x;bg_{-(b, \rho)}(x))=\sup_{l< l^*}L_-(l;b g_{-(b, \rho)}(x))$ for any $x< s_-$, $L_-(x;bg_{-(b, \rho)}(x)) \to -\infty$ when $x\to -\infty$, 
and the function $L_-(x;bg_{-(b, \rho)}(x))$ is increasing iff $g_{-(b, \rho)}$ is decreasing;
\item $L_+(x;bg_{+(b, \rho)}(x))=\inf_{l>l^*}L_+(l;b g_{+(b, \rho)}(x))$ for any $x>s_+$, $L_+(x;bg_{+(b, \rho)}(x)) \to +\infty$ when $x\to +\infty$, and the function $L_+(x;bg_{+(b, \rho)}(x))$ is increasing iff $g_{+(b, \rho)}$ is increasing.
\end{itemize}
\end{proposition}

	The proof is provided in \Cref{proof-of-proposition}. We display a typical plot of $g_{-(b, \rho)}$ and $g_{+(b, \rho)}$ in \Cref{FigNoArbSVI_76_0}.
   
 \begin{figure}
    	\centering
    	\includegraphics[width=.9\textwidth]{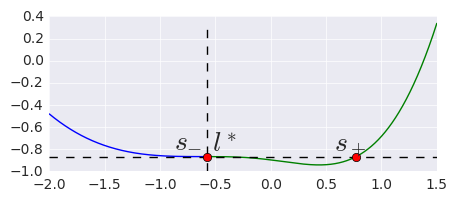}
    	\caption{Typical plot of the functions $g_{\pm(b, \rho)}$ with $b=\frac{2}{3}$ and $\rho=\frac{1}{2}$. The horizontal dotted line corresponds to the level $-b\sqrt{1-\rho^2}$.}
    	\label{FigNoArbSVI_76_0}
 \end{figure}
    
	This proposition has in turn two important corollaries:

\begin{corollary}
Assume (A1) and (B1). There is a unique $(l_-, l_+)$ such that $l_-< s_-,\ l_+>s_+$ and $\alpha = b g_{-(b, \rho)}(l_-)= b g_{+(b, \rho)}(l_+)$. The interval $I_{\alpha,b,\rho}$ is non-empty iff $L_-(l_-;\alpha)< L_+(l_+;\alpha)$.
In this case the distance between $L_+(l_+;\alpha)$ and $L_-(l_-;\alpha)$ increases with $\alpha$.
\end{corollary}

	\begin{proof}
This follows directly from the previous analysis: increasing $\alpha$, the functions $g_{\pm(b, \rho)}$ increase so the corresponding $l_-< s_-$ decreases while $l_+>s_+$ increases. In turn, the function $L_+(l_+;bg_{+(b, \rho)}(l_+))$ increases and the function $L_-(l_-;bg_{-(b, \rho)}(l_-))$ decreases. Note that $l_-< s_-$ and $l_+>s_+$ because $\alpha>-b\sqrt{1-\rho^2}$ from (A1).
We can also use the fact that
\begin{align*}
\frac{d}{d\alpha}(L_+(l_+;\alpha)-L_-(l_-;\alpha)) &= L'_+(l_+)\frac{d}{d\alpha}l_+ - L'_-(l_-)\frac{d}{d\alpha}l_- + \partial_{\alpha}L_+(l_+;\alpha) - \partial_{\alpha}L_-(l_-;\alpha) \\ &= \partial_{\alpha}L_+(l_+;\alpha) - \partial_{\alpha}L_-(l_-;\alpha)
\end{align*}
where $l_+$ and $l_-$ are functions of $\alpha$ given by $\alpha = bg_{+(b, \rho)}(l_+)=bg_{-(b, \rho)}(l_-)$. Now, the RHS is equal to $2\bigl(\frac{1}{N'(l_+)}-\frac{1}{N'(l_-)}-\frac{1}{2}\bigr)$ and since $\frac{1}{N'(l_+)}>\frac{1}{2}$ and $-\frac{1}{N'(l_-)}>\frac{1}{2}$, the previous quantity is greater than $1$.
\end{proof}

Let $F(b, \rho)$ denote the unique value of $\alpha$ such that $L_+(l_+;\alpha)=L_-(l_-;\alpha)$ if there exists such a value for $\alpha>-b\sqrt{1-\rho^2}$, otherwise let $F(b,\rho)=-b\sqrt{1-\rho^2}$. Then  $L_+(l_+;\alpha)>L_-(l_-;\alpha)$ if and only if $\alpha > F(b,\rho)$. In other words we define $F(b,\rho)$ as:
\begin{equation*}
	F(b,\rho):= \inf\{\alpha\mid L_+(l_+;\alpha)>L_-(l_-;\alpha)\} \lor -b \sqrt{1-\rho^2}
\end{equation*}
under the assumptions (A1) and (B1). We name $F$ the Fukasawa threshold of SVI.

	\Cref{FigNoArbSVI_86_0} shows:

\begin{itemize}
\itemsep1pt\parskip0pt\parsep0pt
\item
  in blue the function $l_-\to L_-(l_-;bg_{-(b, \rho)}(l_-))$ with $l_-< s_-$ where $s_-$ is the point at which $g_{-(b, \rho)}$ is equal to $-\sqrt{1-\rho^2}$;
\item
  in green the corresponding value of
  $l_-\to L_+(l_+(bg_{-(b, \rho)}(l_-));bg_{-(b, \rho)}(l_-))$.
\end{itemize}

	\begin{figure}
		\centering
		\includegraphics[width=.9\textwidth]{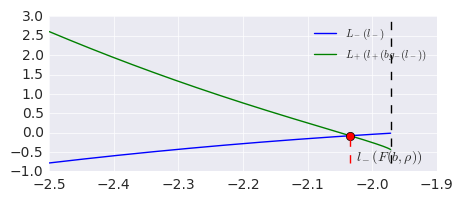}
		\caption{Plot of $L_-(l_-)$ and $L_+(l_+(bg_{-(b, \rho)}(l_-)))$ as functions of $l_-$, with $b=\frac{25}{21}$ and $\rho=\frac{2}{5}$. The vertical dotted line corresponds to the level $-b\sqrt{1-\rho^2}$.}
		\label{FigNoArbSVI_86_0}
	\end{figure}

	The following corollary gives an easy criterion of existence of a
Butterfly arbitrage:

	\begin{corollary} Assume  (A1) and (B1).
If $\alpha\leq F(b,\rho)$ then for every choice of $\mu$ and $\sigma$, the SVI model does not satisfy the Fukasawa conditions.
\end{corollary}

	\subsubsection{Study of the Fukasawa
threshold under (B1)}\label{study-of-the-fukasawa-threshold}

	In the previous section we showed that the difference $L_+(l_+;bg_{+(b, \rho)}(l_+))-L_-(l_-;bg_{-(b, \rho)}(l_-))$ goes
to infinity when increasing
$\alpha = bg_{+(b, \rho)}(l_+)=bg_{-(b, \rho)}(l_-)$ to infinity, so there
exists $\bar{\alpha}$ such that the interval for $\mu$ in non-empty; from
the previous corollaries for each $\alpha>\bar{\alpha}$ the interval for $\mu$ is also in non-empty.
Decreasing $\alpha$, we could bump into two situations:

\begin{itemize}
\itemsep1pt\parskip0pt\parsep0pt
\item
  $\alpha$ reaches the value $F(b,\rho)>-b\sqrt{1-\rho^2}$ for which
  $L_+(l_+;F(b,\rho))=L_-(l_-;F(b,\rho))$;
\item
  $\alpha$ reaches the value $F(b,\rho)=-b\sqrt{1-\rho^2}$. In such case
  $l_\pm=s_\pm$.
\end{itemize}

Our simulations suggest that the first scenario always occurs.

	Could we prove this? In this respect we can observe the following: it is
equivalent to prove that
$L_+(s_+;-b\sqrt{1-\rho^2})< L_-(s_-;-b\sqrt{1-\rho^2})$.

	If $s_+=l^*$ then $L_+(s_+;-b\sqrt{1-\rho^2}) = -l^*$ and the function
$L_+(l_+;bg_{+(b, \rho)}(l_+))$ is increasing. It follows that the
function $L_-(l_-;bg_{-(b, \rho)}(l_-))$ cannot be increasing and
$s_-< l^*$. We should show that $L_-(s_-;-b\sqrt{1-\rho^2})>-l^*$.

When $s_-=l^*$ then $L_-(s_-;-b\sqrt{1-\rho^2}) = -l^*$ and the function
$L_-(l_-;bg_{-(b, \rho)}(l_-))$ is increasing. Again, the function
$L_+(l_+;bg_{+(b, \rho)}(l_+))$ cannot be increasing so $s_+> l^*$.
In this case we should prove that $L_+(s_+;-b\sqrt{1-\rho^2})<-l^*$.

In the final case when both $g_{\pm(b, \rho)}$ have a minimum, it is enough to prove
$L_-(s_-;-b\sqrt{1-\rho^2})>-l^*$ and $L_+(s_+;-b\sqrt{1-\rho^2})<-l^*$.

So to sum up, it would remain to prove that when $g_{-(b, \rho)}$ (or $g_{+(b, \rho)}$) has a minimum, it holds
$L_-(s_-;-b\sqrt{1-\rho^2})>-l^*$ (or $L_+(s_+;-b\sqrt{1-\rho^2})<-l^*$)
to obtain the result in each case. We did not manage to conclude along
those lines though.

	\begin{remark} We don't know whether $F(b,\rho)>-b\sqrt{1-\rho^2}$ but we conjecture it. Indeed we prove in \Cref{computation-of-fb0} that there is a closed formula for $F(b,0)$ which satisfies $F(b,0)>-b$; the statement $F(b,\rho)>-b\sqrt{1-\rho^2}$ can be also assessed numerically.
\end{remark}

	\subsubsection{Symmetries}\label{symmetries}

	We can exploit the symmetry property of $N$ with respect to $\rho$ in
order to restrict the required computations to the function $L_-$ only.

	Indeed
$N(l;\alpha,b,\rho) = N(-l;\alpha,b,-\rho), N'(l;b,\rho) = -N'(-l;b,-\rho)$ and
$N''(l;b)=N''(-l;b)$. This brings to the consideration that
\[L_-(l;\alpha,b,\rho) = -L_+(-l;\alpha,b,-\rho),\quad L_+(l;\alpha,b,\rho) = -L_-(-l;\alpha,b,-\rho),\]
so that
\begin{align*}
\inf_{l>l^*(\rho)}L_+(l;\alpha,b,\rho) &= -\sup_{l>l^*(\rho)}L_-(-l;\alpha,b,-\rho) \\&= -\sup_{l<-l^*(\rho)}L_-(l;\alpha,b,-\rho) \\&= -\sup_{l< l^*(-\rho)}L_-(l;\alpha,b,-\rho)
\end{align*}
so $L_+(l_+(\alpha,b,\rho);\alpha,b,\rho) = -L_-(l_-(\alpha,b,-\rho);\alpha,b,-\rho)$.

	Note that $l_+(\alpha,b,\rho)$ is the unique $l>l^*(\rho)$ such that
$L'_+(l;\alpha,b,\rho)=0$ while $l_-(\alpha,b,-\rho)$ is the unique
$l< l^*(-\rho)$ such that $L'_-(l;\alpha,b,-\rho)=0$. Since
$L'_+(l;\alpha,b,\rho)=L'_-(-l;\alpha,b,-\rho)$ and
$-l_-(\alpha,b,-\rho)>-l^*(-\rho)=l^*(\rho)$, then
$l_+(\alpha,b,\rho)=-l_-(\alpha,b,-\rho)$.

	\begin{lemma}
Assume (A1) and (B1). Then:
\begin{itemize}
\item $L_+(l_+(\alpha,b,\rho);\alpha,b,\rho) = -L_-(l_-(\alpha,b,-\rho);\alpha,b,-\rho)$;
\item $l_+(\alpha,b,\rho)=-l_-(\alpha,b,-\rho)$;
\item $I_{\alpha,b,\rho} = \bigl]L_-(l_-(\alpha,b,\rho);\alpha,b,\rho), -L_-(l_-(\alpha,b,-\rho);\alpha,b,-\rho)\bigr[$.
\end{itemize}

\end{lemma}

	From the above equations we also have
$g_{+(b, \rho)}(l) = g_{-(b, -\rho)}(-l)$ so with easy arguments one gets
$s_+(b,\rho) = -s_-(b,-\rho)$.

	\subsubsection{The cases (B2), (B3) and (B4)}\label{the-case-b1-rho2-or-b1rho2}

	Assume (B2) or (B4). Using the same definitions and following
the proof of \Cref{PropLmgm}, we obtain that $g_{-(b, \rho)}(l)$ is
increasing. Now since $g_{-(b, \rho)}$ is increasing on $]-\infty, l^*]$
and since $g_{-(b, \rho)}(l^*)=-\sqrt{1-\rho^2}$, it follows that there
is no solution to the equation $g_{-(b, \rho)}(l_-)=\frac{\alpha}b$. In this
case so, the supremum of $L_-$ is attained at $-\infty$ and it is
$-\frac{\alpha}{2}$. Under (B3) or (B4), for symmetrical reasons $L_+$ attains its infimum $\frac{\alpha}2$ at $\infty$.

Under (B2), $L_+$ reaches its infimum in $]l^*,+\infty[$ while under (B3), $L_-$ reaches its supremum in $]-\infty,l^*[$. Finally under (B4),
$I_{\alpha,2,0}=\bigl]-\frac{\alpha}{2},\frac{\alpha}{2}\bigr[$.

We can extend the definition of the Fukasawa threshold to the cases (B2), (B3) and (B4):
\begin{itemize}
	\item under (B2), $F(b, \rho)$ denotes the unique value of $\alpha$ such that $L_-(l_-(\alpha,b,-\rho);\alpha,b,-\rho)=\frac{\alpha}{2}$ if there exists such a value for $\alpha>-b\sqrt{1-\rho^2}$, otherwise $F(b,\rho)=-b\sqrt{1-\rho^2}$: $F(b,\rho):= \inf\{\alpha\mid L_+(l_+;\alpha)>-\frac{\alpha}{2}\} \lor -b \sqrt{1-\rho^2}$;
	\item under (B3), $F(b, \rho)$ denotes the unique value of $\alpha$ such that $L_-(l_-(\alpha,b,\rho);\alpha,b,\rho)=\frac{\alpha}{2}$ if there exists such a value for $\alpha>-b\sqrt{1-\rho^2}$, otherwise $F(b,\rho)=-b\sqrt{1-\rho^2}$: $F(b,\rho):= \inf\{\alpha\mid \frac{\alpha}{2}>L_-(l_-;\alpha)\} \lor -b \sqrt{1-\rho^2}$;
	\item under (B4), $F(2,0):=0$.
\end{itemize}
From \Cref{LemmaaPositive}, under cases (B2) and (B3) it holds $F(b,\rho)< 0$.

	\subsection{Conclusion}\label{sectionconclusion}

	We can now state the full characterization of the Fukasawa necessary no
arbitrage conditions for SVI:

	\begin{theorem}[SVI parameters $(\alpha,b,\rho,\mu,\sigma)$ fulfilling Fukasawa necessary no arbitrage conditions]\label{TheoFuk}
Assume (A1). Then:
\begin{itemize}
\item under (B1), $F(b,\rho)< 0$ and the interval $I_{\alpha,b,\rho} = \bigl]L_-(l_-(\alpha,b,\rho);\alpha,b,\rho),\linebreak -L_-(l_-(\alpha,b,-\rho);\alpha,b,-\rho)\bigr[$ is non-empty iff $\alpha> F(b,\rho)$;

\item under (B2) (resp. (B3)), $F(b,\rho)< 0$ and the interval $I_{\alpha,b,\rho} = \bigl]-\frac{\alpha}{2}, -L_-(l_-(\alpha,b,-\rho);\alpha,b,-\rho)\bigr[$ (resp. $I_{\alpha,b,\rho} = \bigl]L_-(l_-(\alpha,b,\rho);\alpha,b,\rho),\frac{\alpha}{2}\bigr[$) is non-empty iff $\alpha> F(b,\rho)$;

\item under (B4), the interval $I_{\alpha,2,0}=\bigl]-\frac{\alpha}{2},\frac{\alpha}{2}\bigr[$ is non-empty iff $\alpha>F(2,0)=0$.
\end{itemize}

In every case, the Fukasawa conditions are satisfied iff $\mu\in I_{\alpha,b,\rho}$.
\end{theorem}

	Except for $F(2,0)$, the result $F(b,\rho)$ negative holds even in the
case $F(b,\rho)>-b\sqrt{1-\rho^2}$ because we have proven that for
$\alpha\geq0$ the interval for $\mu$ is always non-empty. In terms of the usual SVI parameters the conditions translate into
$\frac{a}{\sigma}>F(b,\rho)$ and
$\frac{m}{\sigma}\in I_{\frac{a}{\sigma},b,\rho}$.

	Is the existence of the Fukasawa threshold surprising? We would say no:
indeed the values of $\alpha$ too close to the lower bound
$-b \sqrt{1-\rho^2}$ correspond to values of the smile too close to
zero, and this will lead to an arbitrage as discussed in \cref{smiles-vanishing-at-some-point}, so
that one even expects that $F(b,\rho)>-b\sqrt{1-\rho^2}$.

	The explanation of the range constraint for $\mu$ is less intuitive to us;
we would say that it results from the geometrical constraint that the
Fukasawa conditions impose on the shape of SVI, as follows from our
computations.

	\subsection{Numerics}\label{numerics}

	\paragraph{$F(b,\rho)$ at a fixed $b$}\label{fbrho-at-a-fixed-b}

	We plot in \Cref{FigNoArbSVI_113_0} the Fukasawa threshold at fixed $b=\frac{1}{2}$ as a function of $\rho$.

	\begin{figure}[H]
		\centering
		\includegraphics[width=.8\textwidth]{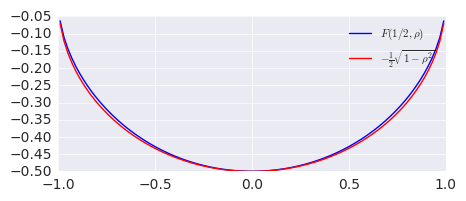}
		\caption{Plot of $F(b,\rho)$ as a function of $\rho$, with $b=\frac{1}{2}$.}
		\label{FigNoArbSVI_113_0}
	\end{figure}
    
	The graph is symmetric with respect to $\rho$ because $F(b,\rho)$ is the
value of $\alpha$ such that the difference between
$L_+(l_+(\alpha,b,\rho);\alpha,b,\rho)$ and $L_-(l_-(\alpha,b,\rho);\alpha,b,\rho)$ is null, where
$bg_{\pm(b, \rho)}(l_\pm(\alpha,b,\rho)) = \alpha$. But\linebreak
$L_+(l_+(\alpha,b,\rho);\alpha,b,\rho) = -L_-(l_-(\alpha,b,-\rho);\alpha,b,-\rho)$ so we
look for $\alpha$ such that
\[L_-(l_-(\alpha,b,-\rho);\alpha,b,-\rho)+L_-(l_-(\alpha,b,\rho);\alpha,b,\rho)=0\]
and this is symmetric with respect to $\rho$.

	The red line is the level $\alpha=-b\sqrt{1-\rho^2}$ and it again confirms
our hypothesis that $F(b,\rho)>-b\sqrt{1-\rho^2}$.

	From the previous graph, it seems that $F(b,\rho)$ has monotonicity of the same sign as $\rho$.
	\newpage

	\paragraph{$F(b,\rho)$ at fixed $\rho$ as a function of
$b$}\label{fbrho-at-fixed-rho-as-a-function-of-b}

	In \Cref{FigNoArbSVI_119_0} we plot the Fukasawa threshold at fixed $\rho=\frac{1}{5}$ as a function of $b$.

	\begin{figure}[H]
		\centering
		\includegraphics[width=.8\textwidth]{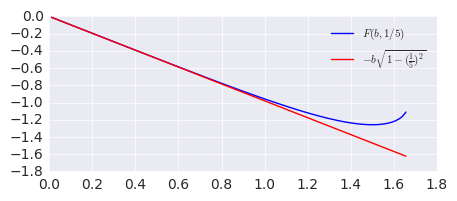}
		\caption{Plot of $F(b,\rho)$ as a function of $b$, with $b=\frac{1}{5}$.}
		\label{FigNoArbSVI_119_0}
	\end{figure}

	\paragraph{$L_-(l_-(F(b,\rho),b,\rho);F(b,\rho),b,\rho)$ and
$L_-(l_-(F(b,\rho),b,-\rho);F(b,\rho),b,-\rho)$ as functions of
$\rho$}\label{lux5f-lux5f-fbrhobrhofbrhobrho-and-lux5f-lux5f-fbrhob-rhofbrhob-rho-as-functions-of-rho}

	The following \Cref{FigNoArbSVI_124_0} shows in blue the function
$L_-(l_-(F(b,\rho),b,\rho);F(b,\rho),b,\rho)$ (denoted for brevity as
$L_-(F(b,\rho),\rho)$) with respect to $\rho$ while in green the
function \linebreak$L_-(l_-(F(b,\rho),b,-\rho);F(b,\rho),b,-\rho)$ (or
$L_-(F(b,\rho),-\rho)$) with respect to $\rho$. The fixed value for $b$
is $\frac{3}{5}$.

	\begin{figure}[H]
		\centering
		\includegraphics[width=.8\textwidth]{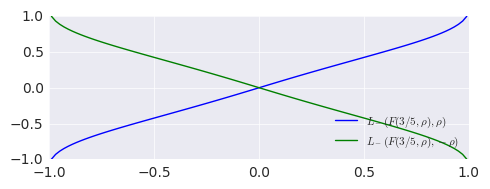}
		\caption{Plot of $L_-(l_-(F(b,\rho),\rho)$ and $L_-(F(b,\rho),-\rho)$ as functions of $\rho$, with $b=\frac{3}{5}$.}
		\label{FigNoArbSVI_124_0}
	\end{figure}
    
	This graph also shows in blue the value of the two bounds for $\mu$ when they shrink to one point. Note that for $\rho=0$ this is $0$ for every
$b$, while it depends on $b$ for the other values of $\rho$.

	The function $\rho\to L_-(l_-(F(b,\rho),b,\rho);F(b,\rho),b,\rho)$ is
odd due to the symmetry of $\rho\to F(b,\rho)$. Furthermore, from the
graph it seems that $\rho$ and
$L_-(l_-(F(b,\rho),b,\rho);F(b,\rho),b,\rho)$ have the same sign.

	\paragraph{$L_-(l_-(F(b,\rho),b,\rho);F(b,\rho),b,\rho)$ as a function
of $b$}\label{lux5f-lux5f-fbrhobrhofbrhobrho-as-a-function-of-b}

	\Cref{FigNoArbSVI_130_0} shows the function
\linebreak$L_-(l_-(F(b,\rho),b,\rho);F(b,\rho),b,\rho)$ (denoted as
$L_-(F(b,\rho),\rho)$) with respect to $b$. Here we fix
$\rho=\frac{1}{2}$.

	\begin{figure}[H]
		\centering
		\includegraphics[width=.8\textwidth]{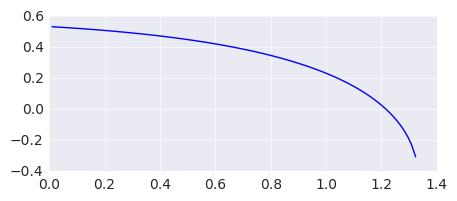}
		\caption{Plot of $L_-(l_-(F(b,\rho),\rho)$ as a function of $b$, with $\rho=\frac{1}{2}$.}
		\label{FigNoArbSVI_130_0}
	\end{figure}
    
	\subsection{Algorithm}\label{firstalgorithm}

	We can \emph{parametrize} the normalized SVI parameters satisfying the
Fukasawa conditions as follows:
	\begin{enumerate}
\def\labelenumi{\arabic{enumi}.}
\itemsep1pt\parskip0pt\parsep0pt
\item
  choose $\rho \in ]-1,1[$ and $b$ positive such that $b(1\pm\rho)\leq 2$ by choosing $b' \in ]0,1]$
  and setting $b=b'\frac{2}{1+|\rho|}$;
\item
  compute numerically $F(b,\rho)$, and parametrize $\alpha$ by setting
  $\alpha=F(b,\rho)+u$ for positive $u$;
\item
  compute numerically $(L_-, L_+)$ for this value of $u$,
  and parametrize $\mu$ by setting
  $\mu = \frac{(1+q)}2 L_+ +\frac{(1-q)}2L_-$ for $q \in ]-1,1[$.
\end{enumerate}

	The values in point 3 can be computed using the same functions employed to
find $F(b,\rho)$, indeed it is sufficient to evaluate
$L_-(l_-(\alpha,b,\rho);\alpha,b,\rho)$ and $-L_-(l_-(\alpha,b,-\rho);\alpha,b,-\rho)$.

	If we are interested only by a test that a given parameter satisfies the
Fukasawa conditions, we have the corresponding waterfall of failure
possibilities that we define as follows:
\begin{enumerate}
\def\labelenumi{\arabic{enumi}.}
\itemsep1pt\parskip0pt\parsep0pt
\item
  $b(1-\rho)>2$ or $b(1+\rho) > 2$: \emph{failure of type 1};
  otherwise:
\item
  $\alpha\leq F(b,\rho)$: \emph{failure of type 2}; otherwise:
\item
  $\mu$ not in $I_{\alpha,b,\rho}$: \emph{failure of type 3}.
\end{enumerate}

	\subsubsection{Application to Axel Vogt
parameters}\label{application-to-axel-vogt-parameters}

	The so-called \emph{Axel Vogt example} (cf \cite{gatheral2014arbitrage})
became the archetypal example of a smile with arbitrage. The $SVI$
parameters are\linebreak
$(a,b,\rho,m,\sigma) = (-0.041,0.1331,0.3060,0.3586,0.4153)$, and they
are known to lead to a Butterfly arbitrage. Do they satisfy the Fukasawa
conditions?

	No, since the respective value for $\mu$ is $0.86347$, while its arbitrage
free interval is $\linebreak]-0.72407, 0.82939[$.

	The Fukasawa conditions are not satisfied because of $\mu$. However
$\alpha=-0.09872$ and $F(b,\rho)=-0.12663$, so $\alpha> F(b,\rho)$ and the
interval for $\mu$ is non-empty. The problem here is due to $\mu$, which is
too large: we face a \emph{failure of type 3}.

\newpage
 \section{No arbitrage domain for SVI}\label{no-arbitrage-domain-for-svi}

	\subsection{Behaviour of the function
$G_2$}\label{behaviour-of-the-function-gux5f2}

	Recall that the function $G_2$ is defined as
\begin{equation}\label{eqG2}
G_2(l) := N''(l)-\frac{N'(l)^2}{2N(l)}
\end{equation}
and that it depends only on $(\alpha,b,\rho)$. As discussed in \cref{mainArg}, $G_2$ is positively proportional to the second derivative of the
\emph{volatility} smile, meaning of $\sqrt{SVI(k)}$. Since the
\emph{variance} smile is convex and asymptotically linear on both sides,
it is expected that $G_2$ will be asymptotically \emph{negative}, while
it is positive around the minimum of the smile. In particular it is
expected that it will have zeros, on both sides of the minimum of the
smile.

	\subsubsection{The zeros of $G_2$}\label{the-zeros-of-gux5f2}

	In this section we prove the following:

	\begin{lemma}[Zeros of $G_2$]
$G_2$ has exactly two zeros $l_1,\ l_2$ which satisfy $l_1 < l^* \wedge 0$ and $l_2 > l^* \lor 0$ such that $G_2(l) > 0 \iff l \in ]l_1,l_2[$. Furthermore, $G_2(l)\to0^-$ for $l\to\pm\infty$.
\end{lemma}

	\begin{proof}
For $l\to\pm\infty$ we have that the first addend behaves as $bl^{-3}$ while the second as $-\frac{b(\rho\pm 1)}{2}l^{-1}$, so $G_2$ behaves as $-\frac{b(\rho\pm 1)}{2}l^{-1}$. This means that $G_2$ goes to $0^-$ as $l\to\pm\infty$.
Since $G_2(l^*) = N''(l^*) > 0$ and $G_2$ is continuous, then there exists an interval $]l_1,l_2[$ containing $l^*$ such that for every $l$ in this interval, $G_2$ is positive. It follows that $G_2$ has at least two zeros.
Deriving, we find the following interesting relationship between $G_2'$ and $G_2$:
$$G_2'(l) = N'''(l)-\frac{N'(l)}{N(l)}G_2(l).$$

We will prove now that this relationship entails that the first zero of $G_2$ is negative. Indeed if $l_1 > 0$ is the first zero of $G_2$, since
\begin{equation}\label{eqN3}
N'''(l) = -\frac{3bl}{(l^2+1)^\frac{5}{2}}
\end{equation}
we have $G_2'(l_1) < 0$, which is not possible because $G_2(l)$ is negative for every $l< l_1$. If $l_1=0$, then $G_2'(0)=0$ but $0$ cannot be a point of local maximum for $G_2$, otherwise there would be a following zero $l_2>0$. In such case, $G_2'(l_2) < 0$ for \cref{eqN3} but having $G_2$ so far negative, it should be increasing in $l_2$. Then $0$ could at most be an inflection point. However,
$$G_2''(l) = N^{iv}(l)-\frac{N'(l)N'''(l)}{N(l)} + \biggl(2\frac{N'(l)^2}{N(l)^2} - \frac{N''(l)}{N(l)}\biggr)G_2(l)$$
so $G_2''(0) = N^{iv}(0) = -3b$, which is negative since $b>0$. Therefore, the first zero $l_1$ of $G_2$ is necessarily negative.
With similar arguments we obtain that the next zero $l_2$ must be non-negative. Suppose $l_2=0$. Then, as before, $G_2'(0)=0$ and $G_2''(0) = -3b < 0$, so it would be a point of local maximum, which is not possible. Then $l_2$ must be positive.

Moreover, there cannot be other zeros for $G_2$. Indeed, suppose $l_3$ was the first zero after $l_2$. Then $l_3> 0$ and from \cref{eqN3} it should be $G_2'(l_2) < 0$ but this cannot be true since $G_2$ is negative in the left neighborhood of $l_3$.

This leads to the conclusion that $G_2$ has exactly two zeros, one positive and the other one negative.
As a consequence, $G_2(0)=b\bigl(1-\frac{b\rho^2}{2(\alpha+b)}\bigr) > 0$. This could have been obtained also from the fact that $\alpha+b\sqrt{1-\rho^2}\geq 0$ due to the positivity of $N$.

Then, we find that $G_2>0$ in $[l^*,0]$ when $\rho\geq 0$ or in $[0,l^*]$ when $\rho < 0$. 
\end{proof}

	Substituting the explicit formulas for $N, N'$ and $N''$ in
\cref{eqG2}, we obtain
\begin{equation*}
G_2(l) = \frac{b}{(l^2+1)^{\frac{3}{2}}} - \frac{b^2(\rho\sqrt{l^2+1} + l)^2}{2(l^2+1)(\alpha+b(\rho l+\sqrt{l^2+1}))}
\end{equation*}
which leads to the remark that
$\frac{G_2(l)}{b}=\tilde{G}_{2,\frac{\alpha}b, \rho}(l)$ where
$\tilde{G}_{2,x, \rho}(l):=\frac{1}{(l^2+1)^{\frac{3}{2}}} - \frac{(\rho\sqrt{l^2+1} + l)^2}{2(l^2+1)(x+(\rho l+\sqrt{l^2+1}))}$,
which reduces in general the study of $G_2$ to the study of a 2-parameters
function.

In order to find the zeros of $G_2$ we should solve
$2\frac{\alpha}b + b(2-l^2)\sqrt{l^2+1} - \rho^2(l^2+1)^{\frac{3}{2}} - 2\rho l^3 = 0$
or equivalently
$2\frac{\alpha}b - 2\rho l^3 = ((\rho^2+1)l^2+\rho^2-2)\sqrt{l^2+1}$.

Note that when $\rho=0$ this equation is explicitly solvable.

	\subsubsection{Plot of a typical $G_2$
function}\label{plot-of-a-typical-gux5f2-function.}

	We plot in \Cref{FigNoArbSVI_161_0} the function $G_2$ for the parameters $\alpha=\frac{1}{10}$,
$b=\frac{1}{2}$, $\rho=-\frac{3}{10}$.

	\begin{figure}[H]
		\centering
		\includegraphics[width=.8\textwidth]{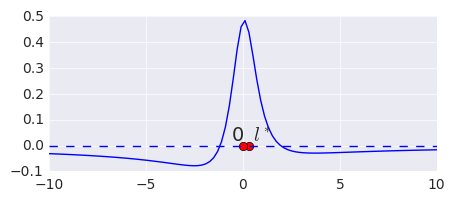}
		\caption{Plot of $G_2$ with $\alpha=\frac{1}{10}$,
			$b=\frac{1}{2}$ and $\rho=-\frac{3}{10}$.}
		\label{FigNoArbSVI_161_0}
	\end{figure}

	\subsection{The final condition on
$\sigma$ under (A1)}\label{the-final-condition-on-sigma}

	We recall that the non-negativity of the Durrleman condition in the case of
SVI amounts to the non-negativity of the function
\begin{equation}\label{eqGG1G2}
G(l) = G_1(l)+\frac{1}{2 \sigma} G_2(l)
\end{equation}
where $G_1$ and $G_2$ do not depend on $\sigma$.

	We have proven that:

\begin{enumerate}
\def\labelenumi{\arabic{enumi}.}
\item
  for every $(\alpha,b,\rho)$ with $b(1\pm\rho)\leq 2$ and
  $\alpha> F(b,\rho)$, where $F(b,\rho)\leq 0$, there exists an interval for
  $\mu$ such that $G_1$ is positive on $\mathbb{R}$ (in fact each factor
  of $G_1$ is positive  on $\mathbb{R}$). Moreover it is necessary that
  the conditions on $(\alpha, b, \rho)$ hold and that $\mu$ lies in this
  interval under no arbitrage.
\item
  for every $(\alpha,b,\rho)$ with $b(1\pm\rho)\leq 2$ there exists
  an interval $]l_1,l_2[$ containing $0$ and $l^*$ such that $G_2(l)>0$
  iff $l\in]l_1,l_2[$.
\end{enumerate}

	We insist here again on the key property brought by the Fukasawa
condition that it is \emph{necessary} that $G_1$ is positive. This
structures a lot the picture; previous to Fukasawa's observation, people
investigating the positivity of $G$ could not assume this. Another
consequence is that under the Fukasawa conditions of \cref{gux5f1-and-the-fukasawa-necessary-condition-for-no-butterfly-arbitrage}, $G$ is
granted to be positive on $[l_1, l_2]$.

	The last step is to exploit the fact that thanks to our
re-parametrization, the dependency of $G$ in $\sigma$ is very simple.
Let $\_$ stand for a fixed set of parameters $(\alpha,b,\rho,\mu)$ fulfilling
the Fukasawa conditions. Then given the fact that $G_2(l)< 0$ for some
$l$, it follows that if $G$ is non-negative everywhere for
$(\_,\sigma)$, then $G$ is also non-negative everywhere for every
$(\_,\tau)$ with $\tau > \sigma$. As a consequence, there exists a
function $\_ \to \sigma^*(\_)$ such that $G$ is non-negative everywhere
for $(\_,\tau)$ iff $\tau \geq \sigma^*(\_)$.

	The value of $\sigma^*$ can be obtained asking the RHS of
\cref{eqGG1G2} to be non-negative, which holds for
$\sigma\geq\sup_l-\frac{G_2(l)}{2G_1(l)}$. Then
\begin{equation*}
\sigma^*(\alpha,b,\rho,\mu) := \sup_{l< l_1 \lor l>l_2}-\frac{G_2(l)}{2G_1(l)}.
\end{equation*}

Since $G_2(l_1^-)=G_2(l_2^+)=0^-$ and $G_2(\pm\infty)=0^-$, the
maximum of $-\frac{G_2(l)}{2G_1(l)}$ for $l< l_1 \lor l>l_2$ is reached
for a finite real value in $]-\infty,l_1]\cup[l_2,+\infty[$.

	We have therefore proven the following:

	\begin{theorem}[Necessary and sufficient no Butterfly arbitrage conditions for SVI under (A1)]\label{FinalTheo}
No Butterfly arbitrage in SVI entails that $G_1$ is positive, which requires $b(1\pm\rho) \leq 2$. Under this condition:
\begin{itemize}
\item each of the factors of the function $G_1$ is positive on $\mathbb R$  if and only if $\alpha> F(b,\rho)$ and $\mu\in I_{\alpha,b,\rho}$;
\item for such $\mu$'s, calling $l_1< 0< l_2$ the only zeros of $G_2$, the function $G$ is positive in $]l_1,l_2[$ for every $\sigma\geq 0$ and the function $G$ is non-negative on $\mathbb{R}$ if and only if $\sigma\geq\sigma^*(\alpha,b,\rho,\mu)$.
\end{itemize}
\end{theorem}

	\subsubsection[Practical computation of sigma star]{Practical computation of
$\sigma^*$}\label{practical-computation-of-sigma}

	Computationally, it would be easier to implement an algorithm with
bounded intervals for $l$. It is enough to substitute $h=\frac{1}{l}$ to
obtain
\[\sigma^*(\alpha,b,\rho,\mu) := \sup_{\frac{1}{l_1} < h < \frac{1}{l_2}}-\frac{G_2(\frac{1}{h})}{2G_1(\frac{1}{h})}.\]

	For $h$ which goes to $0^\pm$, the function $G_2$ goes to $0^-$ while
$G_1$ is always positive under the Fukasawa conditions. So the function
$f\bigl(\frac{1}{h}\bigr)=-\frac{G_2(\frac{1}{h})}{2G_1(\frac{1}{h})}$
goes to $0^+$. This is a point of minimum for $f$ in the interval
$\bigl]\frac{1}{l_1},\frac{1}{l_2}\bigr[$ because here the function is
always positive.

To numerically compute $\sigma^*$ we can use an algorithm which finds
the maximum of $f$ in $\bigl]\frac{1}{l_1},0\bigr[$ and in
$\bigl]0,\frac{1}{l_2}\bigr[$ and then compares the two maxima.

	It can be shown that $f'\bigl(\frac{1}{h}\bigr)$ goes to
$\frac{4b(\rho-1)}{(2-b(\rho-1))(2+b(\rho-1))}< 0$ when $h$ goes to
$0^-$ while it goes to
$\frac{4b(\rho+1)}{(2-b(\rho+1))(2+b(\rho+1))}> 0$ when $h$ goes to
$0^+$. Furthermore, $f'\bigl(\frac{1}{l_1}\bigr)>0$ and
$f'\bigl(\frac{1}{l_2}\bigr)< 0$.

We plot in \Cref{FigNoArbSVI_183_0} the function $f\bigl(\frac{1}{h}\bigr)$ with
$b=\frac{1}{2}$, $\rho=-\frac{3}{10}$, $\alpha=\frac{1}{10}$ and
$\mu=\frac{1}{10}$.

	\begin{figure}
		\centering
		\includegraphics[width=.8\textwidth]{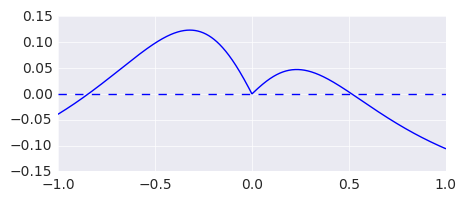}
		\caption{Plot of $f(\frac{1}{h})$ as a function of $h$, with $b=\frac{1}{2}$, $\rho=-\frac{3}{10}$, $\alpha=\frac{1}{10}$ and $\mu=\frac{1}{10}$.}
		\label{FigNoArbSVI_183_0}
	\end{figure}
    
	The function $f\bigl(\frac{1}{h}\bigr)$ seems to have always three
extrema: two points of maximum (one in each interval $\bigl]\frac{1}{l_1},0\bigr[$ and $\bigl]0,\frac{1}{l_2}\bigr[$) and one point of minimum at $0$. The sign of $\rho$ does not imply in which of the two intervals the maximum lies.

	For $\rho=0$ and $\mu=0$ the maxima have the same height, furthermore
the two points of maximum are symmetrical with respect to 0, this last
one is also the point of minimum. This follows from the fact that $G_2$
is symmetric for $\rho=0$ and that when $\mu=0$, also $G_1$ is symmetric.

	Note that we have not proven that there is a single maximum on each side
of $0$. So a strict implementation should take into account the
possibility that there are several ones, and use a global optimizer on
each side. We strongly conjecture that there is in fact a single maximum
on each side.

\subsection{Algorithm under (A1)}\label{algorithm}

We can now complete the algorithms stated for the Fukasawa conditions.
For the parametrization of the no arbitrage domain, we just need to add
the final step which specifies the range of $\sigma$:
\begin{enumerate}
	\def\labelenumi{\arabic{enumi}.}
	\itemsep1pt\parskip0pt\parsep0pt
	\item
	choose $\rho \in ]-1,1[$ and $b$ positive such that $b(1\pm\rho)\leq 2$ by choosing $b' \in ]0,1]$
	and setting $b=b'\frac{2}{1+|\rho|}$;
	\item
	compute numerically $F(b,\rho)$, and parametrize $\alpha$ by setting
	$\alpha=F(b,\rho)+u$ for positive $u$;
	\item
	compute numerically $(L_-, L_+)$ for this value of $u$,
	and parametrize $\mu$ by setting
	$\mu = \frac{(1+q)}2 L_+ +\frac{(1-q)}2L_-$ for $q \in ]-1,1[$;
	\item
	compute numerically $\sigma^*(\alpha, b, \rho, \mu)$, and parametrize $\sigma$ by setting $\sigma=\sigma^*+v$ where $v\geq0$.
\end{enumerate}

The main benefit of this parametrization is that it is eventually a
simple product of intervals:
\[(\rho, b', u, q, v) \in ]-1,1[ \times ]0,1] \times ]0, \infty[ \times ]-1,1[ \times [0, \infty[\]
and this is perfectly suitable to feed optimization algorithms working
with bounds, like the standard ones in the scipy.optimize scientific
library.

A drawback to keep in mind is that sampling this product sub-space in a
uniform way corresponds to a distorted sampling in the initial space.

There again, we can specify an algorithm which decides whether a SVI
parameter lies or not in the no arbitrage domain:
\begin{enumerate}
	\def\labelenumi{\arabic{enumi}.}
	\itemsep1pt\parskip0pt\parsep0pt
	\item
	$b(1-\rho) > 2$ or $b(1+\rho) > 2$: failure of type 1; otherwise:
	\item
	$\alpha\leq F(b,\rho)$: failure of type 2; otherwise:
	\item
	$\mu$ not in $I_{\alpha,b\rho}$: failure of type 3;
	otherwise:
	\item
	$\sigma<\sigma^*$: \emph{failure of type 4}.
\end{enumerate}

	\subsection{The monotonous case (A2) }\label{the-monotonous-case-rho-pm-1}

	In all the previous discussion, we have assumed $|\rho|< 1$ to avoid
singular cases in our computations. What happens when $|\rho| = 1$? We
discuss below the case $\rho=-1$, the case $\rho=1$ follows by symmetry.

In this case the SVI smile is (convex) decreasing, and reaches its
minimum $\alpha$ at infinity, so the domain of $\alpha$ is now $\alpha \geq 0$. Note
that the boundary value $0$ is allowed, unlike in the regular case,
because the implied volatility does not vanish at any finite strike. The
negative slope condition requires $b\leq 1$, and the positive
(rightmost) one is automatically fulfilled.

Regarding the Fukasawa conditions, the proofs in \cref{gux5f1-and-the-fukasawa-necessary-condition-for-no-butterfly-arbitrage} still hold
with the convention that $l^*=+\infty$ so that $N$ is decreasing. The
interval for $\mu$ becomes\linebreak
$I_{\alpha,b,-1}=]L_-(l_-(\alpha,b,-1);\alpha,b,-1),+\infty[$, so exactly equal to
$I_{\alpha,b,\rho}$ with the convention\linebreak $L_-(l_-(\alpha,b,1);\alpha,b,1)=-\infty$.
For $\alpha\geq0$, we have $L_-(l)< 0$ for every $l$ so
\linebreak$L_-(l_-(\alpha,b,-1);\alpha,b,-1)< 0$ also, and this interval always contains
$[0,\infty[$. We can then extend the definition of the Fukasawa threshold to the case $\rho=-1$, putting $F(b,-1)=0$. This implies that the
interval for $\mu$ is non-degenerate even when $\alpha=F(b,-1)=0$.

The function $G_2$ has only one negative zero $l_1$, above which it is
always positive with $G_2(+\infty)=0^+$ while $G_2(-\infty)=0^-$. So
$\sigma^*= \sup_{l< l_1}-\frac{G_2(l)}{2G_1(l)}$.

	\begin{theorem}[Necessary and sufficient no Butterfly arbitrage conditions for SVI, $\rho=-1$]
No Butterfly arbitrage in SVI entails that $G_1$ is positive, which requires $b\leq 1$ and $\alpha \geq 0$. Under these conditions:
\begin{itemize}
\item each of the factors of the function $G_1$ is positive on $\mathbb R$  if and only if $\mu>L_-(l_-;\alpha,b,-1)$;
\item for such $\mu$'s, calling $l_1< 0$ the only zero of $G_2$,  the function $G$ is positive on $]l_1,\infty[$ for every $\sigma\geq 0$ and the function $G$ is non-negative on $\mathbb{R}$ if and only if $\sigma\geq\sigma^*(\alpha,b,-1,\mu)$ where $\sigma^*(\alpha,b,-1,\mu) = \sup_{l< l_1}-\frac{G_2(l)}{2G_1(l)}$.
\end{itemize}

\end{theorem}

	\subsubsection{Application: SVI decreasing to
zero}\label{application-svi-decreasing-to-zero}

	Let us consider the case $\rho=-1$ and $a=0$, so SVI is given by the formula
$SVI(k) = b(-(k-\mu \sigma)+\sqrt{(k-\mu \sigma)^2+\sigma^2})$
with $b\leq 1$.

	Can we compute the lower bound for $\mu$? Consider the equation
$g_{-(b, -1)}(l)=0$ or equivalently from \cref{eqgmexplicit},
$(-\sqrt{l^2+1} + l)\bigr(\sqrt{l^2+1}\bigl(\frac{1}{2}-\frac{b}{4}\bigr) + \frac{bl}{4} \bigr)+1 = 0$.
Simplifying, we obtain $2l(1-b)\sqrt{l^2+1}=2(1-b)l^2-(b+2)$ and
squaring we find the two solutions $l=\pm\frac{b+2}{2\sqrt{3(1-b)}}$
when $b< 1$. The positive one does not solve the initial equation, so
with the notations used in \cref{gux5f1-and-the-fukasawa-necessary-condition-for-no-butterfly-arbitrage}, we finally find
$s_-=-\frac{b+2}{2\sqrt{3(1-b)}}$. If $b=1$, then $s_-=+\infty$. Note
that $s_-$ corresponds to $l_-$ when $\alpha=0$, and we get that
$L_-(l_-(0,b,-1);0,b,-1)=-\sqrt{3(1-b)}$.

So for $\alpha=0$:
\begin{itemize}
\itemsep1pt\parskip0pt\parsep0pt
\item
  the Fukasawa conditions are satisfied if and only if
  $\mu > -\sqrt{3(1-b)}$;
\item
  the unique zero of $G_2$ does not depend on $b$ and is given by
  $l_1=-\frac{1}{\sqrt{3}}$, and the parameters with no arbitrage are
  eventually given by $b\leq 1$, $\mu > -\sqrt{3(1-b)}$,
  $\sigma\geq \sigma^*(0,b,-1,\mu) = \sup_{l< -1/\sqrt{3}}-\frac{G_2(l)}{2G_1(l)}$.
\end{itemize}

\section{Calibration experiments}\label{calibration-experiments}

	Now that we have parametrized the no arbitrage domain, the design of a
calibration algorithm is straightforward:
\begin{enumerate}
\def\labelenumi{\arabic{enumi}.}
\itemsep1pt\parskip0pt\parsep0pt
\item
  choose an objective function;
\item
  choose a starting point policy;
\item
  for the chosen starting points (possibly several of them), run a
  minimization algorithm of the objective function over the no arbitrage
  domain;
\item
  pick up the optimal parameters.
\end{enumerate}

	As objective function, we choose the classical least squares criterion,
which takes as input the differences of the data and model total
variances on the available set of log-forward moneyness. This will
give equal weights to far-from the money points, where the precise value
of the implied volatility, and so the accuracy of the calibration,
matters less, and to close-to the money ones, which is not a desirable
feature: it can be easily patched by adding weights given by the Vegas
(computed once for all with the data points), so that the errors are
more in line with losses, unit-wise. This would moreover stabilize the
calibration from one day to another one, especially on illiquid markets,
as discussed in detail in \cite{nagy2019volatility}.

Now the big question for us is rather whether or not the no arbitrage
constraint will deteriorate the quality of the fit, and we will also
work on model generated data or on index options data which are liquid
ones, whence our choice of a standard non-weighted objective function.

	Regarding the starting point policy, we are not big fans of \emph{smart
guess strategies} which try to compute the best starting point from the
data. Such strategies can work brilliantly in many favorable situations,
yet they might fail heavily on data with low quality (e.g.~due to a
dubious treatment by an internal department), or when faced with new
market behavior and configurations. There is a clear risk of
over-engineering here also. We would be more confident by using a
\emph{set} (with small cardinality) of starting points, possibly
produced by a machine learning algorithm duly trained on the markets in
scope. We implement a very basic version of this idea, which picks up
uniformly generated points within the hyperrectangle of the no arbitrage
domain, irrespective of the data.

	The scipy function here used is the `least\_squares' which lies in the
optimize library. The method used is the `dogbox', which handles bounds.
The tolerances regarding the change of the cost function (`ftol'), the
change of the independent variables (`xtol') and the norm of the
gradient (`gtol') are all set at the Python numpy machine epsilon. The
maximum number of function evaluations (`max\_nfev') is set at $1000$.

	Even though the arbitrage region does not impose an upper bound for $\alpha$
and $\sigma$, we choose arbitrary ones. In particular, we ask
\[\sigma\leq\max{\biggl(\frac{|k_0|}{r},\frac{|k_N|}{r}, 1.5\sigma^*\biggr)}\]
with $r$ as parameter to be chosen by the user (default value equal to
$0.1$). This bound is related to the fact that when
$\frac{|k_i|}{\sigma}$ is below a threshold $r$, then the smile is
almost flat and this causes uncertainty on the parameters to be chosen.

The upper bound for $\alpha$ is left to be chosen by the user. For the index
option data we set $\alpha< 1$ since it is enough to achieve a very good
fit, while for the model generated data, in order to have an almost
perfect calibration, the upper bound actually depends on the $\alpha$
parameter used to generate data. We set in every case $\alpha< 3$, since we
know a priori that all the data are generated with $\alpha$ lower than $3$.

	We provide below our calibration results on model generated data and
then on market data.

	\subsection{On model data}\label{on-model-data}

	To check the robustness of the algorithm we firstly run it on data
generated by arbitrary SVI parameters with no arbitrage, and on the Axel
Vogt parameters. We take a vector of 13 log-forward strikes taken from
Table 3.2 of \cite{ferhati2020robust}.
    
	The parameters chosen for each of the graphs in \Cref{FigNoArbSVI_216_0} are
arbitrage-free. The red and the blue lines, which represent the total
variances generated from the arbitrary parameters and the total
variances obtained from the calibrated parameters respectively, overlap.

\begin{figure}[H]
	\centering
	\includegraphics[width=.9\textwidth]{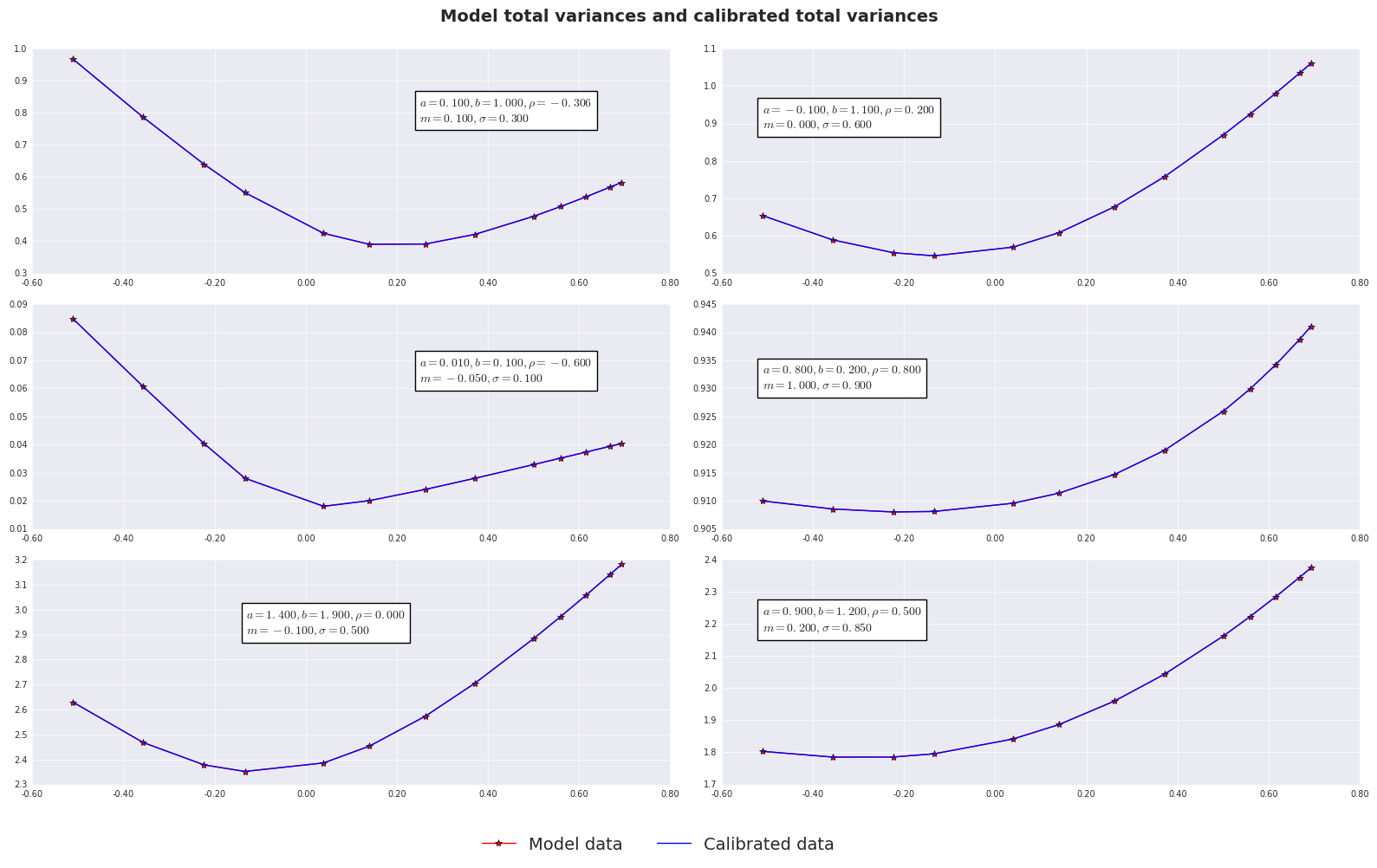}
	\caption{Model total variances generated by arbitrage-free parameters (in red) and calibrated total variances (in blue).}
	\label{FigNoArbSVI_216_0}
\end{figure}

The fact that the fit is excellent can be seen by the Frobenius relative
errors in \Cref{tableFrobeniusTotVar}.

\begin{table}[H]
	\caption{Frobenius relative errors for the total variances with arbitrage-free parameters calibrated on model total variances.}
	\label{tableFrobeniusTotVar}
\begin{center}
 \begin{tabular}{lrrrrrr}
\toprule
{} &   $a$ &  $b$ &  $\rho$ &   $m$ &  $\sigma$ &  Relative error $(\times 10^{-16})$ \\
\midrule
0 &  0.10 &  1.0 &  -0.306 &  0.10 &      0.30 &                                1.48 \\
1 & -0.10 &  1.1 &   0.200 &  0.00 &      0.60 &                                1.63 \\
2 &  0.01 &  0.1 &  -0.600 & -0.05 &      0.10 &                                2.30 \\
3 &  0.80 &  0.2 &   0.800 &  1.00 &      0.90 &                                1.77 \\
4 &  1.40 &  1.9 &   0.000 & -0.10 &      0.50 &                                2.35 \\
5 &  0.90 &  1.2 &   0.500 &  0.20 &      0.85 &                                2.25 \\
\bottomrule
\end{tabular}
\end{center}
\end{table}

	Furthermore, also the Frobenius relative error on the parameters is low (\Cref{tableFrobeniusParameters}).
This means that the algorithm is robust and recovers the original data.

\begin{table}[H]
	\caption{Frobenius relative errors for the parameters calibrated on model total variances.}
	\label{tableFrobeniusParameters}
\begin{center}
 \begin{tabular}{lrrrrrr}
\toprule
{} &   $a$ &  $b$ &  $\rho$ &   $m$ &  $\sigma$ &  Relative error $(\times 10^{-14})$ \\
\midrule
0 &  0.10 &  1.0 &  -0.306 &  0.10 &      0.30 &                                0.10 \\
1 & -0.10 &  1.1 &   0.200 &  0.00 &      0.60 &                                0.30 \\
2 &  0.01 &  0.1 &  -0.600 & -0.05 &      0.10 &                                0.06 \\
3 &  0.80 &  0.2 &   0.800 &  1.00 &      0.90 &                               20.00 \\
4 &  1.40 &  1.9 &   0.000 & -0.10 &      0.50 &                                0.10 \\
5 &  0.90 &  1.2 &   0.500 &  0.20 &      0.85 &                                3.00 \\
\bottomrule
\end{tabular}
\end{center}
\end{table}

	\subsubsection{Axel Vogt parameters}\label{axel-vogt-parameters}

	For a matter of completeness we run our algorithm on the notorious Axel
Vogt parameters, which lead to an arbitrage SVI. The original and the
calibrated parameters are reported in \Cref{tableAxelVogtParameters} while the graphs of the original and arbitrage-free total variances are shown in \Cref{FigNoArbSVI_229_0}.
    
\begin{table}[H]
	\caption{Axel Vogt parameters vs best fitting no arbitrage.}
	\label{tableAxelVogtParameters}
\begin{center}
 \begin{tabular}{llllll}
\toprule
{} &        $a$ &       $b$ &    $\rho$ &       $m$ &  $\sigma$ \\
\midrule
Original   &     -0.041 &    0.1331 &     0.306 &    0.3586 &    0.4153 \\
Calibrated & -0.0198444 &  0.102745 &  0.180754 &  0.266125 &  0.310459 \\
\bottomrule
\end{tabular}
\end{center}
\end{table}

\begin{figure}[H]
	\centering
	\includegraphics[width=.8\textwidth]{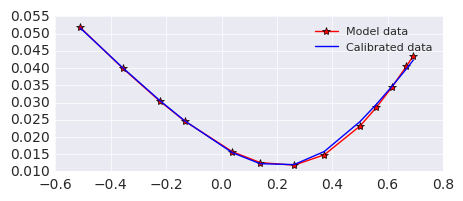}
	\caption{Total variances generated by the Axel Vogt parameters (in red) and total variances with arbitrage-free parameters (in blue).}
	\label{FigNoArbSVI_229_0}
\end{figure}

	Of course, the calibration is not perfect as in the previous case and
the Frobenius error between the Axel Vogt total variances and the non
arbitrage SVI corresponding total variances is $2.15\%$.

	We compare the function $g$ defined in \cref{eqgSVI} with the
original Axel Vogt parameters and the same function with the new
arbitrage-free parameters in \Cref{FigNoArbSVI_234_0}.

\begin{figure}[H]
	\centering
	\includegraphics[width=.8\textwidth]{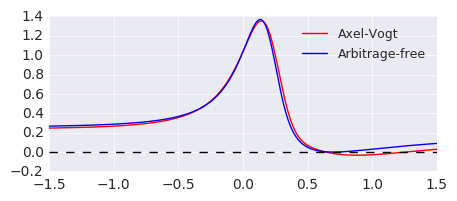}
	\caption{Plot of the functions $g$ with the the Axel Vogt parameters (in red) and with the arbitrage-free parameters (in blue).}
	\label{FigNoArbSVI_234_0}
\end{figure}
    
	From the plot it can be seen that the function $g$ with the new
arbitrage-free parameters can be very close to zero, but it is always
positive.

	In the following study, we compare the results obtained with the new
arbitrage-free parameters and the ones with the parameters described in
Example 5.1 of \cite{gatheral2014arbitrage}, which are also arbitrage
free. \Cref{FigNoArbSVI_242_0} shows that the fit of our new parameters is
better than the one of Gatheral and Jacquier.

\begin{figure}[H]
	\centering
	\includegraphics[width=.9\textwidth]{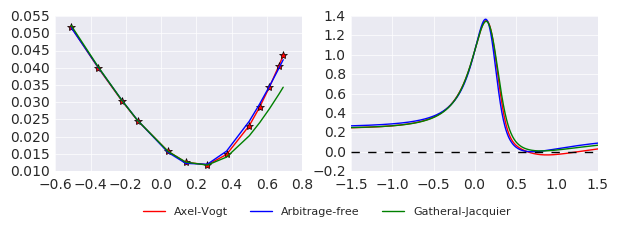}
	\caption{On the left, plot of the total variances generated by the Axel Vogt parameters (in red), the total variances with the arbitrage-free parameters (in blue) and the total variances with the Gatheral-Jacquier parameters (in green). On the right, plot of the function $g$ with the the Axel Vogt parameters (in red), with the arbitrage-free parameters (in blue) and with the Gatheral-Jacquier parameters (in green).}
	\label{FigNoArbSVI_242_0}
\end{figure}

    In \Cref{tableGatheralParameters} we compare the relative errors on the total
variances for the two sets of arbitrage-free parameters.

\begin{table}[H]
	\caption{Frobenius relative errors for the total variances with arbitrage-free parameters vs Gatheral-Jacquier parameters.}
	\label{tableGatheralParameters}
\begin{center} 
 \begin{tabular}{lllllll}
\toprule
{} &        $a$ &       $b$ &    $\rho$ &       $m$ &  $\sigma$ & Relative error \\
\midrule
Arbitrage-Free    & -0.0198444 &  0.102745 &  0.180754 &  0.266125 &  0.310459 &          0.022 \\
Gatheral-Jacquier & -0.0305199 &  0.102717 &  0.100718 &  0.272344 &  0.412398 &          0.133 \\
\bottomrule
\end{tabular}
\end{center}
\end{table}

	\subsection{On data from CBOE}\label{on-data-from-cboe}

	We now turn to market data. We work with market data of good quality
bought from the CBOE datastore by Zeliade. They cover daily files for
the DJX, SPX500 and NDX equity indices, with bid and ask prices.

	To obtain implied total variances from the prices, we operate the
classical treatment of inferring the discount factor and forward values
at each option maturity by performing a linear regression of the (mid)
Call minus Put prices with respect to the strike. Since the markets
under study are very liquid, the fit is excellent and the residual error is
extremely small.

Then, given the discount factor and forward values for each maturity,
we are able (after working out the exact maturity of each contract from
its code, if not provided explicitly) to compute the implied
volatilities, for the Bid and Ask prices.

We feed the objective function with the implied volatility corresponding
to the mid price, and plot below the implied volatilities for the
calibrated model and the bid and ask market data. Results are reported in \Cref{FigNoArbSVI_255_1,FigNoArbSVI_258_1,FigNoArbSVI_261_1}.

\begin{figure}[H]
	\centering
	\includegraphics[width=.9\textwidth]{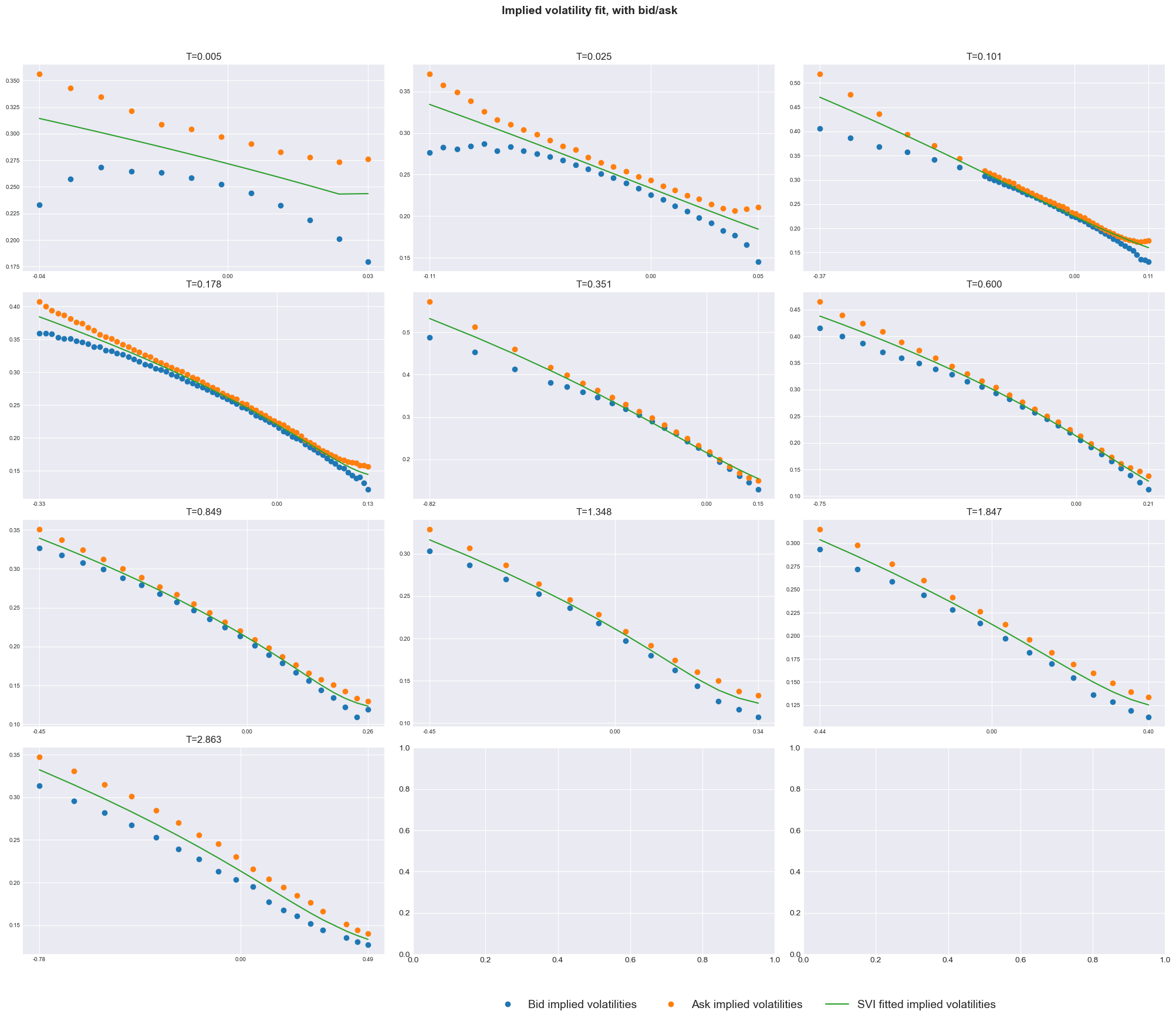}
	\caption{Fitted implied volatilities with arbitrage-free parameters (in red) and bid (in blue) and ask (in green) implied volatilities for the DJX index.}
	\label{FigNoArbSVI_255_1}
\end{figure}

\begin{figure}[H]
	\centering
	\includegraphics[width=.9\textwidth]{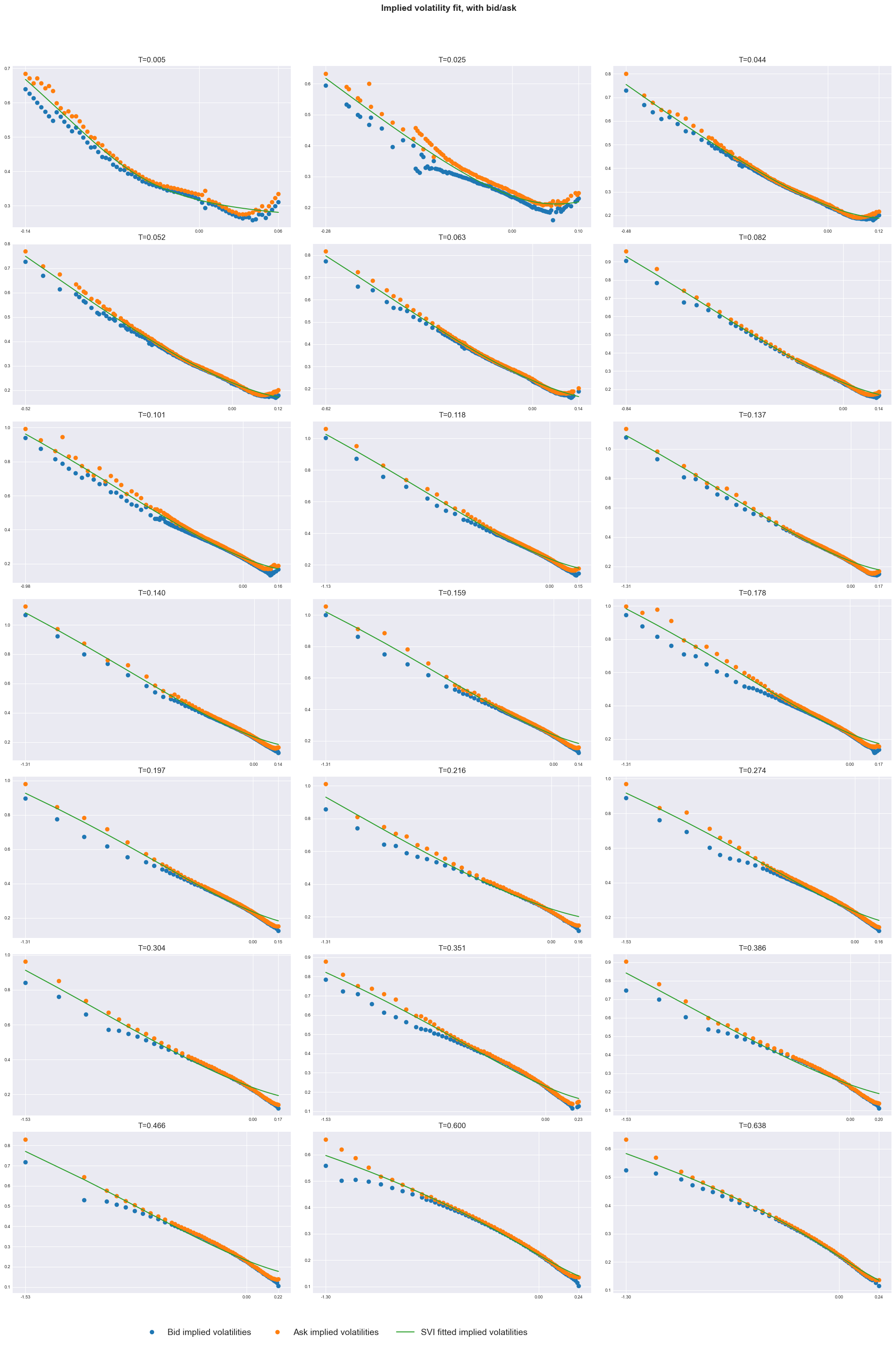}
	\caption{Fitted implied volatilities with arbitrage-free parameters (in red) and bid (in blue) and ask (in green) implied volatilities for the SPX500 index.}
	\label{FigNoArbSVI_258_1}
\end{figure}

\begin{figure}[H]
	\centering
	\includegraphics[width=.9\textwidth]{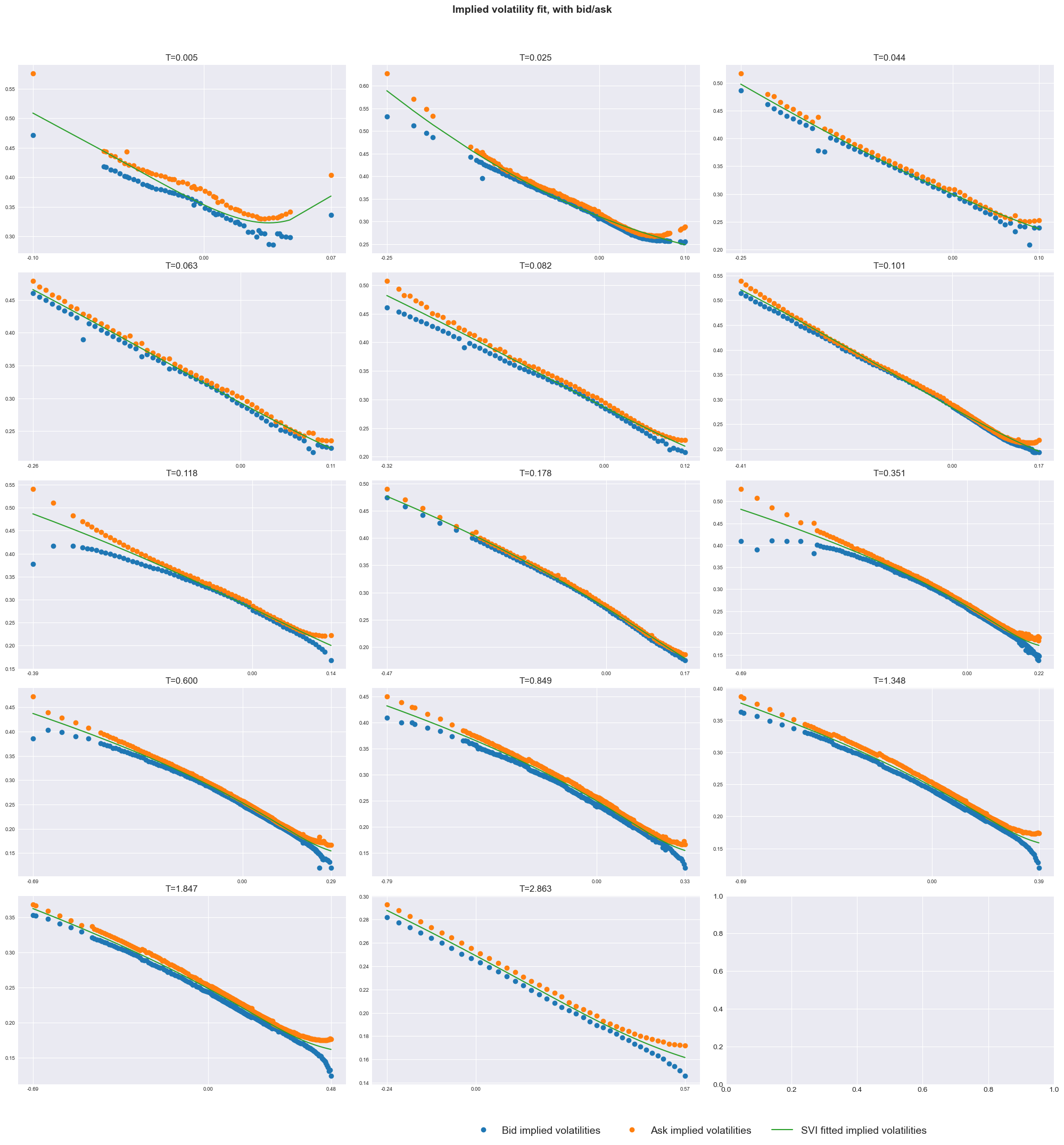}
	\caption{Fitted implied volatilities with arbitrage-free parameters (in red) and bid (in blue) and ask (in green) implied volatilities for the NDX index.}
	\label{FigNoArbSVI_261_1}
\end{figure}

	\subsection{Discussion}\label{discussion}

	From our experiments we draw several positive conclusions:
\begin{itemize}
\itemsep1pt\parskip0pt\parsep0pt
\item
  the quality fit is excellent, and there is no big loss resulting from
  the no arbitrage constraint;
\item
  the implementation we have designed seems sufficiently robust in
  practice; of course such a statement should be re-assessed
  continuously;
\item
  the payload of the root finding algorithms used to compute the
  Fukasawa threshold and the bounds for $\mu$ and $\sigma$ is not an
  issue, the calibration is still reasonably fast on a basic chip; the
  average for each maturity for the DJX data is $51.598$ seconds, for
  the SPX data $36.490$ seconds and for the NDX data $44.900$ seconds.
\end{itemize}

Of course, there is room for improvement, at least at the level of the
starting point strategy. One could also think of pre-computing the
numerical functions computed on the fly, or to design once for all
explicit proxies for them, which would speed massively the execution of
the algorithm.

 \section{Conclusion}\label{conclusion}

	Fukasawa's remark that the inverse of $d_1$ and $d_2$ functions of the
Black-Scholes formula have to be increasing under no Butterfly
arbitrage, paired with the natural rescaling of the SVI parameters which
consists in scaling $a$ and $m$ by $\sigma$, allow us to fully describe
the domain of no Butterfly arbitrage for SVI.

The no Butterfly arbitrage domain can be parametrized as an hyperrectangle, with 2 downstream algorithms of practical importance: one for
checking that a SVI parameter lies or not in the no arbitrage domain,
and the other one to effectively perform a calibration. Three functions
have to be computed numerically by resorting to root-finding type
algorithms; due to the fact that our careful mathematical analysis
provided safe bracketing intervals for those functions, this can be
achieved in a very quick manner. We provide calibration results on model
and market data, the latter showing that there is no loss of fit quality
due to imposing the no arbitrage constraint.

This analysis settles one important issue in the SVI saga. Other ones
are still pending, like the study of sub-SVI parametrizations with 4
parameters instead of 5, in the spirit of SSVI (which has 3 parameters
slice-wise), which could display more parameter stability than SVI and a better fit
quality than SSVI; and also the question of the characterization of no
Calendar Spread arbitrage for two SVI slices corresponding to different
maturities.

\newpage

\appendix\section[Proof of Proposition 5.5]{Proof of \Cref{PropLmgm}}\label{proof-of-proposition}

	\begin{proof}
Observe that at the point $l^*$, $\rho\sqrt{l^2+1} + l=0$ and also after computations, $\rho l+\sqrt{l^2+1} = \sqrt{1-\rho^2}$, so we have $g_{\pm(b, \rho)}(l^*)=-\sqrt{1-\rho^2}$.

We have $\frac{d}{d\alpha}L_\pm(l_\pm) = L_\pm'(l_\pm)\frac{d}{d\alpha}l_\pm + \partial_{\alpha}L_\pm(l_\pm) = \partial_{\alpha}L_\pm(l_\pm)$. Deriving \cref{eqLmp} with respect to $\alpha$, we find $\partial_{\alpha}L_\pm(l_\pm) = 2\Bigl(\frac{1}{N'(l_\pm)}\mp\frac{1}{4}\Bigr)$.

Since $N'(l)> 0$ iff $l> l^*$ and $4\mp N'>0$, then $\partial_{\alpha}L_-(l_-)< 0$ and $\partial_{\alpha}L_+(l_+)>0$. So the function $\alpha\to L_-(l_-,\alpha)$ is decreasing while $\alpha\to L_+(l_+,\alpha)$ is increasing. It means that the bounds for $\mu$ are an increasing family of sets (possibly empty) parametrized by $\alpha$. 
Consider the lower bound, so $l< l^*$. We can write the expression for $g_{-(b, \rho)}$ in another way. We have
$$L'_-(l) = 1 + \frac{N'(l)}{2} - \frac{2N''(l)}{N'(l)^2}(\alpha + lN'(l) + N''(l)(l^2+1)).$$
Evaluating this in $l_-$, the LHS becomes $0$ and we can isolate $\alpha$, obtaining
\begin{equation}\label{eqgmimplicit}
g_{-(b, \rho)}(l) = \frac{1}{b}\biggl(\frac{N'(l)^2}{2N''(l)}\biggl(1+\frac{N'(l)}{2}\biggr) - lN'(l) - N''(l)(l^2+1)\biggr).
\end{equation}

From this expression, we get the derivative of $g_{-(b, \rho)}$ such as $g'_{-(b, \rho)}(l) = \frac{N'(l)^2}{4b}\bigl(3 - \frac{N'''(l)}{N''(l)^2}(N'(l) + 2)\bigr)$, which is positive iff the second factor is positive. Substituting with the explicit expressions, we find that this holds iff $\frac{3}{b}\sqrt{l^2+1}\bigl(l(2+b\rho) + b\sqrt{l^2+1}\bigr) > 0$ or equivalently $-l(2+b\rho) < b\sqrt{l^2+1}$.

Note that since $b(1-\rho)\leq2$, then $2+b\rho>0$. If $\rho< 0$ then the equation is true for $l\geq 0$. For negative $l$s we can square, obtaining that it holds iff $(b^2(\rho^2-1)+ 4b\rho + 4)l^2 < b^2$. For $b(1-\rho)<2$, the coefficient of $l^2$ is positive, so the inequality holds iff $l>-\frac{b}{\sqrt{b^2(\rho^2-1)+ 4b\rho + 4}} := m_-$. Since $\rho$ is negative, $m_- < l^*$. So in this case $g_{-(b, \rho)}(l)$ is increasing iff $l >m_-$.

If $\rho\geq 0$, we proceed in a similar way taking the square and obtaining that, if $b(1-\rho)<2$, the inequality holds iff $l>-\frac{b}{\sqrt{b^2(\rho^2-1)+ 4b\rho + 4}} := m_-$. If $b\leq\frac{2\rho}{1-\rho^2}$, then $m_-\geq l^*$ and $g_-$ is always decreasing. Otherwise if $b>\frac{2\rho}{1-\rho^2}$, then $m_-< l^*$ and $g_{-(b, \rho)}$ is increasing iff $l> m_-$.
We can write $\alpha$ as a function of $l_-$, indeed $\alpha = bg_{-(b, \rho)}(l_-)$. This function has the same monotonicity as $g_{-(b, \rho)}$.

We obtain from the previous analysis that the function $x\to L_-(x;bg_{-(b, \rho)}(x))$ is:

- increasing iff $x< m_-$ when $b>\frac{2\rho}{1-\rho^2}$;

- increasing for every $x< l^*$ when $b\leq\frac{2\rho}{1-\rho^2}$.

Note that $N(l) = \alpha + lN'(l) + N''(l)(l^2+1)$. Using \cref{eqLmp} and substituting $\alpha$ with $bg_{-(b, \rho)}(x)$ considered as in \cref{eqgmimplicit}, we obtain
\begin{equation}\label{eqLmbgm}
L_-(x;bg_{-(b, \rho)}(x)) = \frac{N'(x)^2}{N''(x)}\biggl(1+\frac{N'(x)}{2}\biggr)\biggl(\frac{1}{N'(x)} + \frac{1}{4}\biggr) - x.
\end{equation}

From here it can be seen that $L_-(x;bg_{-(b, \rho)}(x))$ goes to $-l^*$ when $x$ goes to $l^{*-}$.
Similarly, we can do all the equivalent computations for $L_+$. First, the function $g_{+(b, \rho)}$ can be re-written as
\begin{equation*}
g_{+(b, \rho)}(l) = \frac{1}{b}\biggl(\frac{N'(l)^2}{2N''(l)}\biggl(1-\frac{N'(l)}{2}\biggr) - lN'(l) - N''(l)(l^2+1)\biggr)
\end{equation*}
while
$$L_+(x;bg_{+(b, \rho)}(x)) = \frac{N'(x)^2}{N''(x)}\biggl(1-\frac{N'(x)}{2}\biggr)\biggl(\frac{1}{N'(x)} - \frac{1}{4}\biggr) - x$$
and even in this case $L_+(x;bg_{+(b, \rho)}(x))$ goes to $-l^*$ when $x$ goes to $l^{*+}$.
We can study the monotonicity of $g_{+(b, \rho)}$, obtaining $g'_{+(b, \rho)}(l) = \frac{N'(l)^2}{4b}\bigl(-3 + \frac{N'''(l)}{N''(l)^2}(N'(l) - 2)\bigr)$.

Considering the second factor and substituting with the explicit expressions, the latter quantity is positive iff $-\frac{3}{b}\sqrt{l^2+1}\bigl(-l(2-b\rho) + b\sqrt{l^2+1}\bigr) > 0$ or equivalently $l(2-b\rho) > b\sqrt{l^2+1}$.

Here, since $b(1+\rho)\leq 2$, then $2-b\rho>0$. If $\rho> 0$ then the equation is false for $l\leq 0$. For positive $l$s we can square, obtaining that it holds iff $(b^2(\rho^2-1)- 4b\rho + 4)l^2 > b^2$. For $b(1+\rho)< 2$, the coefficient of $l^2$ is positive, so the inequality holds iff $l>\frac{b}{\sqrt{b^2(\rho^2-1)- 4b\rho + 4}} := m_+$. Since $\rho$ is positive, $m_+ > l^*$. So in this case $g_{+(b, \rho)}(l)$ is increasing iff $l >m_+$.

If $\rho\leq 0$, we proceed in a similar way taking the square and obtaining that, if $b(1+\rho)< 2$, the inequality holds iff $l>\frac{b}{\sqrt{b^2(\rho^2-1)- 4b\rho + 4}} := m_+$. If $b\leq-\frac{2\rho}{1-\rho^2}$, then $m_+\leq l^*$ and $g_{+(b, \rho)}$ is always increasing. Otherwise if $b>-\frac{2\rho}{1-\rho^2}$, then $m_+> l^*$ and $g_{+(b, \rho)}$ is increasing iff $l> m_+$.
Remember that the function $\alpha\to L_+(l_+,\alpha)$ is increasing. To recap, the function $x\to L_+(x;bg_{+(b, \rho)}(x))$ is:

- increasing iff $x> m_+$ when $b>-\frac{2\rho}{1-\rho^2}$;

- increasing for every $x> l^*$ when $b\leq-\frac{2\rho}{1-\rho^2}$.

If $b\leq-\frac{2\rho}{1-\rho^2}$ then $\rho< 0$ and $b>\frac{2\rho}{1-\rho^2}$ while if $b\leq\frac{2\rho}{1-\rho^2}$ then $\rho>0$ and $b>-\frac{2\rho}{1-\rho^2}$. This means that $L_+(x;bg_{+(b, \rho)}(x))$ and $L_-(x;bg_{-(b, \rho)}(x))$ cannot be both monotonous.

The last statement of the proposition is a direct consequence to the fact that \linebreak$\frac{d}{dx}L_\pm(x;bg_{\pm(b, \rho)}(x)) = \partial_{\alpha}L_\pm(x;bg_{\pm(b, \rho)}(x))bg'_{\pm(b, \rho)}(x)$ where $\partial_{\alpha}L_-(x)< 0$ and $\partial_{\alpha}L_+(x)>0$.
\end{proof}

\section{Computation of $F(b,0)$}\label{computation-of-fb0}

	In this appendix we compute $F(b,0)$ and prove that $F(b,0)>-b$.

	With $\rho = 0$ we have $l^*=0$ and
\begin{align*}
N &= \alpha + b\sqrt{l^2+1},\\
N' &= \frac{bl}{\sqrt{l^2+1}},\\
N'' &= \frac{b}{(l^2+1)^\frac{3}{2}}.
\end{align*}

	Consider the particular case $b=2$. Then we have already shown
$F(2,0)=0$, which is greater than $-2$.

	Consider $b\neq2$. Since $b>\frac{2\rho}{1-\rho^2}=0$, then the function
$l_-\to L_-(l_-;bg_{-(b,0)}(l_-))$ is increasing iff $l_-< m_-$ where
$m_-=-\frac{b}{\sqrt{4-b^2}}$. Furthermore the Fukasawa interval for $\mu$ is equal to
$I_{\alpha,b,0} = \bigl]L_-(l_-(\alpha,b,0);\alpha,b,0), -L_-(l_-(\alpha,b,0);\alpha,b,0)\bigr[$
so it is symmetrical with respect to $0$. The Fukasawa threshold
$F(b,0)$ is then the solution to $L_-(l_-(F(b,0),b,0);F(b,0),b,0)=0$.

	From equation \cref{eqLmbgm} we obtain
\[L_-(l_-;bg_{-(b,0)}(l_-)) = b\frac{l_-^2}{2}\Bigl(2\sqrt{l_-^2+1}+bl_-\Bigr)\Biggl(\frac{\sqrt{l_-^2+1}}{bl_-}+\frac{1}{4}\Biggr)-l_-.\]
For $l_-< 0$, this expression is equal to $0$ iff
$(8+b^2)l=-6b\sqrt{l^2+1}$ and so iff $l_-$ equals
$l_-^*:=-\frac{6b}{\sqrt{b^4-20b^2+64}}$. Then
\[F(b,0) = bg_{-(b,0)}\biggl(-\frac{6b}{\sqrt{b^4-20b^2+64}}\biggr)\]
where $g_{-(b,0)}(l) = \frac{l^2}{4}(2\sqrt{l^2+1}+bl) - \sqrt{l^2+1}$.

	We now need to prove $g_{-(b,0)}(l_-^*)>-1$ or equivalently
$l_-^*< s_-$. From the expression of $g_{-(b,0)}$, we immediately find
that $s_-$ satisfies $2(l^2-2)\sqrt{l^2+1}=-bl^3-4$, so we look for a
negative root such that $\frac{-bl^3-4}{l^2-2}>0$. This happens iff $l$
lies outside the interval
$\Bigl[\Bigl(-\frac{4}{b}\Bigr)^\frac{1}{3},-\sqrt{2}\Bigr]$ if
$b\leq\sqrt{2}$, or outside the interval
$\Bigl[-\sqrt{2},\Bigl(-\frac{4}{b}\Bigr)^\frac{1}{3}\Bigr]$ if
$b>\sqrt{2}$. Squaring the previous equation and simplifying by $l^3$ we
find $(4-b^2)l^3-12l-8b=0$. Call $P_b(l)$ the LHS.

At $0$, this polynomial and its derivative are negative. Its local
maximum is at $-\frac{2}{\sqrt{4-b^2}}$ and its value at this point is
$\frac{16}{\sqrt{4-b^2}}-8b$ which is always positive. So the polynomial
has two negative roots and a positive one.

We can observe that
$P_{\sqrt{2}}(-\sqrt{2})=P_{\sqrt{2}}\Bigl(\Bigl(-\frac{4}{b}\Bigr)^\frac{1}{3}\Bigl)=0$
with
\begin{itemize}
\itemsep1pt\parskip0pt\parsep0pt
\item
  $P_b(-\sqrt{2}) = 2\sqrt{2}(b-\sqrt{2})^2> 0$,
\item
  $P_b\Bigl(\Bigl(-\frac{4}{b}\Bigr)^\frac{1}{3}\Bigr) = -\frac{4}{b}\Bigl(b^2-3(2b)^{\frac{2}{3}}+4\Bigr)< 0$.
\end{itemize}

Then if $b<\sqrt{2}$ the root of interest $s_-$ is the second negative
root of the polynomial while if $b\geq\sqrt{2}$ it is the first negative
root.

	The value of the polynomial in $l^*_-$ is
\[- \frac{8 b ((b^{2} - 16)^2 - 36 \sqrt{b^{4} - 20 b^{2} + 64})}{(b^{2} - 16)^2}\]
which is positive iff $b< \tilde{b}$ where
$\sqrt{\frac{8}{5}} <\tilde{b} < \sqrt{2}$. The derivative of the
polynomial evaluated in $l^*_-$ is $\frac{24(5 b^{2}-8)}{16-b^{2}}$,
which is positive iff $b>\sqrt{\frac{8}{5}}$.

	Then:
\begin{itemize}
\itemsep1pt\parskip0pt\parsep0pt
\item
  if $b\leq \tilde{b}$ the polynomial is positive in $l^*_-$ and $s_-$
  is its second root, so $l^*_-< s_-$;
\item
  if $\tilde{b} < b < \sqrt{2}$ the polynomial is negative in $l^*_-$
  while its derivative is positive and $s_-$ is its second root, so
  $l^*_-< s_-$;
\item
  finally if $b \geq \sqrt{2}$ the polynomial is negative with a positive
  derivative in $l^*_-$ so even if $s_-$ is now its first root we have
  $l^*_-< s_-$.
\end{itemize}

\bibliographystyle{plain}
\bibliography{biblioNoArb}

    \end{document}